\documentclass{article}
\usepackage[table]{xcolor}

\usepackage{graphicx} 
\usepackage{subcaption} 
\usepackage{tikz} 
\usetikzlibrary{calc}
\usetikzlibrary{quantikz} 
\tikzset{
  header node/.style = {
    font          = \Large,
    text depth    = +7pt,
    fill          = white,
    draw},
  header/.style = {%
    inner ysep = +1.5em,
    append after command = {
      \pgfextra{\let\TikZlastnode\tikzlastnode}
      node [header node] (header-\TikZlastnode) at (\TikZlastnode.north) {#1}
    }
  },
}

\usepackage[T1]{fontenc}
\usepackage[utf8]{inputenc}
\usepackage{csquotes}
\usepackage[english]{babel}

\usepackage{booktabs} 
\usepackage{diagbox}
\usepackage{multirow}
\usepackage{floatrow} 
\usepackage{placeins} 

\usepackage[colorlinks]{hyperref} 

\usepackage{braket} 
\usepackage{amsmath}
\usepackage{amssymb}
\usepackage{amsthm} 
\usepackage{thmtools}
\usepackage{thm-restate}
\usepackage{dsfont} 

\usepackage[]{algorithm} 
\usepackage{algpseudocode} 

\usepackage{enumitem}

\usepackage{authblk} 
\usepackage{fancyhdr}		
\usepackage{multicol} 

\usepackage{rotating}

\usepackage[left=1.5cm,right=1.5cm,bottom=2cm, top=2cm]{geometry} 

\title{Exponential Qubit Reduction in Optimization for Financial Transaction Settlement}

\author[1,2]{Elias X. Huber$^\dagger~$} 

\author[2]{Benjamin Y. L. Tan$^\ddagger~$}

\author[4]{Paul R. Griffin}

\author[2,3, 5]{Dimitris G. Angelakis$^*~$}

\affil[1]{
ETH Zürich, D-CHAB Interdisciplinary Sciences, Vladimir-Prelog-Weg 1-5/10, 8093 Zurich
}
\affil[2]{
Centre for Quantum Technologies, National University of Singapore, 3 Science Drive 2, Singapore 117543
}
\affil[3]{
AngelQ Quantum Computing, 531A Upper Cross Street, 04-95 Hong Lim Complex, Singapore 051531
}
\affil[4]{
School of Computing and Information Systems, Singapore Management University, 80 Stamford Road, Singapore 178902
}

\affil[5]{School of Electrical and Computer Engineering, Technical University of Crete, Chania, Greece 73100}

\affil[$\dagger$]{\href{mailto:elixaver@gmail.com}{elixaver@gmail.com}}
\affil[$\ddagger$]{\href{mailto:b.tan@u.nus.edu}{b.tan@u.nus.edu}}
\affil[*]{\href{mailto:dimitris.angelakis@gmail.com}{dimitris.angelakis@gmail.com}}

\date{\today}

\usepackage[%
  backend=bibtex      
 ,style=numeric-comp  
 ,sorting=none        
 ,sortcites=true      
 ,block=none
 ,indexing=false
 ,citereset=none
 ,isbn=true
 ,url=true
 ,doi=true            
 ,natbib=true         
]{biblatex}
\setlist[description]{leftmargin=1cm,labelindent=1cm}

\newcommand{\myvec}[1]{\underline{\mathbf{#1}}}
\newcommand{\mthet}{\myvec{\theta}}

\newcommand{\ra}[1]{\renewcommand{\arraystretch}{#1}}
\newcommand{\tablefsA}{\small}
\newcommand{\tablefsB}{\tiny}
\DeclareMathOperator*{\argmax}{arg\,max}
\DeclareMathOperator*{\argmin}{arg\,min}

\definecolor{goodgreen}{RGB}{168, 240, 180}
\definecolor{badred}{RGB}{247, 117, 77}
\definecolor{mediumyellow}{RGB}{229, 232, 70}
\definecolorset{gray/rgb/hsb/cmyk}{}{}%
 {black,0/0,0,0/0,0,0/0,0,0,1;%
  darkgray,.25/.25,.25,.25/0,0,.25/0,0,0,.75;%
  gray,.5/.5,.5,.5/0,0,.5/0,0,0,.5;%
  lightgray,.75/.75,.75,.75/0,0,.75/0,0,0,.25;%
  white,1/1,1,1/0,0,1/0,0,0,0}
\definecolor{matplotlibbrown}{HTML}{A52A2A}

\begin{document}

\maketitle

\begin{abstract}
    We extend the qubit-efficient encoding presented in~\cite{tan_qubit-efficient_2021} and apply it to instances of the financial transaction settlement problem constructed from data provided by a regulated financial exchange. Our methods are directly applicable to any QUBO problem with linear inequality constraints.
    Our extension of previously proposed methods consists of a simplification in varying the number of qubits used to encode correlations as well as a new class of variational circuits which incorporate symmetries thereby reducing sampling overhead, improving numerical stability and recovering the expression of the cost objective as a Hermitian observable. We also propose optimality-preserving methods to reduce variance in real-world data and substitute continuous slack variables.
    We benchmark our methods against standard QAOA for problems consisting of 16 transactions and obtain competitive results. Our newly proposed variational ansatz performs best overall. We demonstrate tackling problems with 128 transactions on real quantum hardware, exceeding previous results bounded by NISQ hardware by almost two orders of magnitude. 
\end{abstract}

\renewcommand{\abstractname}{Keywords}

\begin{abstract} 
Quantum Computing, Quantum Optimization, NISQ, QUBO, Mixed binary optimization, Quantum Finance, Qubit reduction
\end{abstract}

\renewcommand{\abstractname}{Abstract}

\section{Introduction}
Provable asymptotic advantages of quantum computing over classical algorithms have been shown in the \emph{fault-tolerant} regime (\cite{grover_fast_1996, shor_algorithms_1994}) and \emph{quantum computational supremacy} (\cite{harrow_quantum_2018}) has been claimed experimentally in circuit sampling tasks (\cite{arute_quantum_2019, zhong_quantum_2020, madsen_quantum_2022}).\footnote{Although some problem instances have later been shown to be classically simulable, e.g.~\cite{pan_solving_2022}.} Methods that promise to extend these computational advantages to \underline{relevant} problems with \underline{available} \emph{noisy intermediate scale quantum} (NISQ) devices have been an active field of research over the past decade. A recent breakthrough in this regard was achieved by IBM Quantum (\cite{kim_evidence_2023}), claiming evidence for the utility of said NISQ devices by simulating the evolution under an Ising Hamiltonian beyond the reach of standard\footnote{Efficient classical simulation of the experiment was claimed shortly thereafter, e.g.~using tensor networks: \cite{tindall_efficient_2023, kechedzhi_effective_2023, begusic_fast_2023}.} classical simulation methods. Most research to this end of useful NISQ algorithms is concerned with problems in Hamiltonian simulation, machine learning or energy minimization/optimization (\cite{bharti_noisy_2022, wei_nisq_nodate}). This work concerns the latter.


\paragraph{Outline of this paper} In the introduction, we give an overview of quantum optimization in NISQ, summarize different approaches to reduce the number of qubits (\ref{subsec:quantum_optimization_intro}) and introduce the transaction settlement problem (\ref{subsec:fin_trs_settlement_intro}). We extend the qubit reduction technique introduced in ~\cite{tan_qubit-efficient_2021} to find approximate solutions to problem instances larger than previously attempted. We outline the mapping used between the quantum state and the binary variables of the problem (\ref{subsec:qubit_compression}) and how the cost can be estimated using this quantum state (\ref{subsec:cost}). We introduce a new variational ansatz derived to incorporating symmetries of the encoding scheme (\ref{subsec:vqa_ansatz}), before concluding with simulation (\ref{subsec:simulation_results}) and quantum hardware (\ref{subsec:hw_ionq_ibmq}) results.

\subsection{Quantum Optimization -- Quadratic Unconstrained Binary Optimization}\label{subsec:quantum_optimization_intro}

The optimization problem we consider in~\ref{subsec:fin_trs_settlement_intro} will generalize \emph{quadratic unconstrained binary optimization} (QUBO) problems, which have the form
\begin{equation}
    \argmin_{\myvec{x} \in \{0,1\}^I} C(\myvec x) = \argmin_{\myvec{x} \in \{0,1\}^I} \myvec x^T Q \myvec x
    \label{eq:QUBO}
\end{equation}
where $I$ is the number of binary entries of the vector $\myvec x$ and $Q$ is any real matrix, $Q \in \mathbb{R}^{I\times I}$.
Finding the vector $\myvec x$ minimizing equ.~\ref{eq:QUBO} for general $Q$ is NP hard (\cite{barahona_computational_1982}). Many combinatorial/graph problems such as MaxCut can be readily mapped to QUBO problems and a wide range of industrial applications is known. This includes training of machine learning models (\cite{date_qubo_2021}) 
and optimization tasks such as assignment problems (\cite{vikstal_applying_2020}), route optimization (\cite{harwood_formulating_2021}) or - the focus of this study - financial transaction settlement (\cite{braine_quantum_2019}).\footnote{For further applications see chapter 2 in \cite{punnen_quadratic_2022}.} This broad applicability and (by benchmarking existing classical solvers) \enquote{verifiable} advantage makes QUBO problems a great test-bed in the search for a useful quantum advantage.

The solution of equation~\ref{eq:QUBO} corresponds to the ground state of an Ising Hamiltonian $H_Q$ on $I$ qubits,
\begin{equation}
    H_Q = \frac{1}{4} \sum_{i,j = 1}^IQ_{ij}(1-\sigma_z^{i})(1-\sigma_z^{j}),
    \label{eq:QUBO_hamiltonian}
\end{equation}
with $\sigma_a^i$ referring to the Pauli operator $a$ on qubit $i$. This allows mapping a QUBO problem on $I$ variables to the problem of finding the ground state of a Hamiltonian on $I$ qubits. We extend equation~\ref{eq:QUBO_hamiltonian} in section~\ref{subsec:qubit_compression} and \ref{subsec:cost} by applying the qubit compression from~\cite{tan_qubit-efficient_2021} to reduce the number of qubits to $O(\text{log}I)$ at the cost of losing the formulation~\ref{eq:QUBO_hamiltonian} as the ground state of a Hermitian operator.
\emph{Quantum solvers} (QS) to the Ising Hamiltonian or more general ground-state problems have been studied extensively. A short overview is given in table~\ref{tab:Qsolver_Annealing_QAOA_VQA_assisted} and the following: 

\paragraph{Annealing} Introduced as early as 1994 (\cite{finnila_quantum_1994}) and inspired by simulated annealing (\cite{kirkpatrick_optimization_1984}), \emph{quantum annealing} aims to find the ground state of $H_Q$ in~\ref{eq:QUBO_hamiltonian} by \emph{adiabatically} transforming $H_{\text{tot}}(t) = s(t) H_Q + (1-s(t))H_m$, with the \emph{mixing Hamiltonian} $H_m = \sum_{j=1}^I\sigma^x_j$, over a time span $t \in [0, t_{\text{end}}]$. Here, $s(t)$ is the annealing schedule, with $s(0)=0$ and $s(t_{\text{end}})=1$. Reading out the state of the annealing device at the end of this transformation yields candidates for the optimal solution $\myvec x$. Annealing devices are not guaranteed to find optimal solutions efficiently and can only implement a limited set of Hamiltonians, often restricted in their connectivity (resulting in limitations on the non-zero entries of $Q$) (\cite{yarkoni_quantum_2022}). Despite these limitations, general-purpose QUBO solvers based on hybrid classical computation and quantum annealing are commercially available with as many as 5000 (1 million) physical nodes (variables, $I$ in equ.~\ref{eq:QUBO}) for D-wave's Advantage$^\text{TM}$ annealer (\cite{inc_d-wave_nodate}).\footnote{Note that it is not public how exactly the quantum annealer is used as a subroutine in this hybrid computation.} 
\paragraph{QAOA} \emph{\textbf{Q}uantum \textbf{A}pproximate \textbf{O}ptimization \textbf{A}lgorithms}~(\cite{farhi_quantum_2014}) can be regarded as implementing a parametrized, \emph{trotterized} version of the quantum annealing schedule on gate-model based quantum computers. The parameterized $p$-layered circuit $e^{-iH_Q\beta_p}e^{-iH_m\gamma_p}...e^{-iH_Q\beta_1}e^{-iH_m\gamma_1}$ is applied to $\ket{+}^{\otimes I}$ and measured in the computational basis. This yields candidate vectors $\myvec x$ by identifying each binary variable with one qubit. The parameters $\{\beta_j, \gamma_j\}$ are classically optimized to minimize equ.~\ref{eq:QUBO} (minimize $\braket{H_Q}$). 
QAOA provides theoretical guarantees in its convergence to the exact solution for $p\to \infty$ given optimal parameters. Yet, implementing the evolution of $H_Q$ and reaching sufficient depth $p$ on NISQ devices can be infeasible in the case of many non-zero entries of $Q$.
\paragraph{Hardware-efficient VQA} In this work, we make use of general \emph{Variational Quantum Algorithms} to minimize a cost estimator (in the context of quantum chemistry often referred to as VQE, \emph{variational quantum eigensolver} (\cite{peruzzo_variational_2014}), and applied beyond Ising Hamiltonians). VQAs are general quantum circuit ansätze parameterized by classical parameters, hence QAOA can be seen as a special case of a VQA. We use the term \emph{hardware-efficient} VQA loosely for ansätze whose gates, number of qubits and circuit depth suit current NISQ devices. Analogously to QAOA, the parameters of the VQA circuit are optimized classically through evaluation of some classical cost function on the measured bit-vector. As we will see later, this cost function does not necessarily correspond to a Hermitian observable. 
VQAs are widely studied in the NISQ era beyond their application to combinatorial optimization problems (\cite{tilly_variational_2022, benedetti_parameterized_2019, mcclean_theory_2016}). 
Challenges, most notably vanishing gradients for expressive circuits (\cite{mcclean_barren_2018, arrasmith_effect_2021, wang_noise-induced_2021}) and remedies (\cite{liu_mitigating_2022, pesah_absence_2021, patti_entanglement_2021, grant_initialization_2019, dborin_matrix_2021, skolik_layerwise_2021, cerezo_cost_2021, schatzki_theoretical_2022, sack_avoiding_2022}) exist aplenty but will not play a central role in this paper. While the generality of VQAs allows for tailored hardware-efficient ansätze which are independent of the problem itself, this comes at the cost of losing the remaining theoretical guarantees of QAOA and adiabatic ground state computation. 

\paragraph{Non-VQA, quantum-assisted solvers} Other quantum algorithms for solving ground-state problems have been proposed in the literature. Examples include quantum-assisted algorithms, often inspired by methods such as Krylov subspace, imaginary time evolution or quantum phase estimation. For example, quantum computers are used to calculate overlaps between quantum states employed in a classical outer optimization loop (\cite{seki_quantum_2021, kyriienko_quantum_2020, bharti_iterative_2021, takeshita_increasing_2020, motta_determining_2020, huggins_non-orthogonal_2020, stair_multireference_2020}). Although some of these approaches are variational in the circuit ansatz, they do not directly correspond to the classical-quantum feedback loop in the VQA setting described above and are beyond the focus of this work.

\medskip

\paragraph{Classical solvers} It should be noted at this point, that approaches using classical computing for tackling QUBO problems exist. Among them\footnote{A more extensive overview can be found in chapter 11 of \cite{punnen_quadratic_2022}.} are general purpose optimization suites such as Gurobi (\cite{gurobi_optimization_llc_gurobi_2023}), CPLEX (\cite{cplex_v12_2009}) or SCIP (\cite{achterberg_scip_2009}) as well as dedicated approximation algorithms such as simulated annealing (\cite{kirkpatrick_optimization_1983}), TABU search (\cite{wang_multilevel_2012})) or the relaxation-based Goemans and Williamson (\cite{goemans_improved_1995}) algorithm which guarantees an approximation ratio of at least 0.878\footnote{Which is optimal for any polynomial-time classical algorithm assuming the unique games conjecture (\cite{khot_optimal_2005}).} for Max-Cut problems.
Due to the NP-hardness of the general problem, all classical solvers are either approximations or have no polynomial worst-case runtime guarantees.

\medskip

\paragraph{NISQ-Limitations} Quantum computers are not expected to break NP-hardness (cf. \cite{bennett_strengths_1997, aaronson_limits_2008} and the lack of any polynomial-time quantum algorithm for an NP-hard problem) and it is often justified to regard quantum approaches to QUBO as heuristics \emph{hoped} to provide practical advantages rather than general purpose solvers with rigorous runtime \emph{and} optimality guarantees.
This makes benchmarking on relevant problem instances paramount in guiding the search for promising quantum algorithms. Yet, most NISQ-era quantum approaches suffer from a combination of
\begin{enumerate}
    \item \textbf{Problem size} limited by the number of available qubits
    \item Constraints on the \textbf{problem class} (connectivity of $Q$)
\end{enumerate}
making a direct application of QS to \emph{relevant} problem instances infeasible on NISQ-devices (\cite{guerreschi_qaoa_2019}). 
While 1.~is a consequence of the limited number of qubits available on NISQ devices, 2.~can be seen as a consequence of noise in the qubit and operations:
Computations become infeasible due to low coherence times and noisy gates paired with often deep circuits (e.g.~arising from the limited lattice-connectivity of devices based on superconducting qubits) upon decomposition into hardware-native gates. 
Constraints on the problem class can also arise from the fundamental design of the algorithm itself. 

How these limitations on problem size and class apply to the different QS is summarized in table~\ref{tab:Qsolver_Annealing_QAOA_VQA_assisted}. Various work has been done to address these challenges. Improved problem embeddings (\cite{date_efficiently_2019}), decomposition (\cite{mitarai_overhead_2021}), compilation and hardware-efficient ansätze are just some approaches to deal with connectivity issues. A wide variety of qubit-reduction methods has been suggested in the quantum optimization and quantum chemistry literature, see table~\ref{tab:qubit-reduction-methods}. 

Proposing a solution to the limitations in table~\ref{tab:Qsolver_Annealing_QAOA_VQA_assisted} and pushing the boundaries of QUBO problems accessible by QS is a central motivation for this work. We give a detailed description of our qubit-reduction method in section~\ref{subsec:qubit_compression}.

\begin{table*}[h]\centering \small
\ra{1.3}
\begin{tabular}{@{}lc>{\columncolor{lightgray!30!white}}c>{\columncolor{lightgray!30!white}}cc@{}}\toprule
    \multirow{3}{*}{\diagbox[width=3cm, height=1.5cm]{\normalsize NISQ \\ \small limitation}{\normalsize Algorithm}} & \multicolumn{4}{c}{\large Quantum solvers/heuristics (QS)}
    \\
\cmidrule{2-5} 
    & \normalsize{Annealing} 
    &\normalsize{QAOA} 
    &\setlength\extrarowheight{-3pt}\begin{tabular}{@{}c@{}}\tiny{(Hardware-efficient)}\\\normalsize{VQA / VQE}  \end{tabular}
    & \normalsize{Quantum assisted solver}
    \\ \midrule
    &&&&\\ 
    \small{\#Variables $I$} 
    & {\setlength\extrarowheight{-3pt}\begin{tabular}{@{}c@{}}\small{\#qubits = $I$}\\ Realized experimentally: $I\leq 5000$\end{tabular}} 
    & \multicolumn{2}{c}{\cellcolor{lightgray!30!white}\setlength\extrarowheight{-3pt}\begin{tabular}{@{}c@{}}\small{\#qubits = $I$}\\ Realized experimentally: $I\leq 127$\end{tabular}} 
    & \small{\#qubits = $I$}\\
    &&&&\\ 
    {\setlength\extrarowheight{-3pt}\begin{tabular}{@{}l@{}}\small{Connectivity: }\\$Q_{ij} \overset{!}{=}0$ for some $i\neq j$\end{tabular}} 
    & {\setlength\extrarowheight{-3pt}\begin{tabular}{@{}c@{}} restricted to device\\connectivity\end{tabular}} 
    & {\setlength\extrarowheight{-3pt}\begin{tabular}{@{}c@{}}problem vs. device \\ connectivity\\$\leftrightarrow$ circuit depth\end{tabular}} 
    & ansatz-dependent
    & {\setlength\extrarowheight{-3pt}\begin{tabular}{@{}c@{}}ansatz-dependent, \\ overlap calculation\end{tabular}} 
    \\
\bottomrule
\\
\small{References}
&\cite{inc_d-wave_nodate}
&\cellcolor{white}\cite{harrigan_quantum_2021, otterbach_unsupervised_2017, pelofske_quantum_nodate, zhu_multi-round_2022, shaydulin_qaoa_2023}
&\cellcolor{white}\cite{braine_quantum_2019, peruzzo_variational_2014, tilly_variational_2022, mcclean_theory_2016}
&\cite{seki_quantum_2021, kyriienko_quantum_2020, bharti_iterative_2021, takeshita_increasing_2020, motta_determining_2020, huggins_non-orthogonal_2020, stair_multireference_2020}\\ 
\end{tabular}
\caption{Overview of QS for ground state problems and NISQ limitations in the \enquote{vanilla} formulations of these approaches. Shown is the relation between physical qubits and the number of variables as well as the impact of the problem connectivity. Grey underlaid will be the focus in this work: Our qubit-efficient encoding makes use of VQA and we will benchmark our results against QAOA.
}
\label{tab:Qsolver_Annealing_QAOA_VQA_assisted}
\end{table*}

\begin{table*}[h]\centering \small
\ra{1.3}
\begin{tabular}{@{}lccc>{\columncolor{lightgray!30!white}}ccc@{}}
\multirow{3}{*}{\diagbox[width=2.5cm, height=1.7cm]{\normalsize Features}{\normalsize Reduction \\ {\small method}}} 
& &\multicolumn{2}{c}{\normalsize Divide \& Conquer} &\cellcolor{white} & &
\\
\cmidrule{3-4} 
& \begin{tabular}{@{}c@{}}\normalsize Problem \\ reduction\end{tabular} & {\normalsize Objective} & {\normalsize Quantum circuit} & \begin{tabular}{@{}c@{}}\normalsize Binary encoding \\ \normalsize qubit compression\end{tabular} & \begin{tabular}{@{}c@{}}\normalsize Qubit efficient \\ \normalsize relaxation\end{tabular} & {\normalsize Qubit-reuse}
\\
\bottomrule
&&&&&&\\
\normalsize{Description}
& \begin{tabular}{@{}c@{}} Pre-processing to \\ simplify problem\end{tabular}
& \multicolumn{2}{c}{\begin{tabular}{@{}c@{}} Split in smaller problems \\ and recombine\end{tabular}}
& \begin{tabular}{@{}c@{}} Replace one-hot enc. \\ by binary enc.\end{tabular}
& \begin{tabular}{@{}c@{}} Map to \\ relaxation\end{tabular}
& \begin{tabular}{@{}c@{}} Mid-circuit \\ msm. \& reset \end{tabular}
\\
&&&&&&\\
\normalsize{QS applicable} & All & All & All gate-based & VQA (QAOA) & (VQA) & VQA 
\\
&&&&&&\\
\begin{tabular}{@{}l@{}}\normalsize Qubit \\ \normalsize reduction\end{tabular} 
& \begin{tabular}{@{}c@{}} problem-\\ dependent \end{tabular} 
& \multicolumn{2}{c}{$\times \frac{1}{c}$ or set to constant}
& exponential
& \begin{tabular}{@{}c@{}} $\times \frac{1}{c}$ / \\ exponential \end{tabular}
& set to constant
\\
&&&&&&\\
{\normalsize Caveats}
& \begin{tabular}{@{}c@{}} No guaranteed \\ reduction\end{tabular}
& \multicolumn{2}{c}{\begin{tabular}{@{}c@{}} Divide/Conquer impacted\\ by connectivity, circuit overhead\end{tabular}}
& msm. overhead 
& \begin{tabular}{@{}c@{}} $c$ $\leftrightarrow$ problem \\ connectivity / \\ connectivity \& \\ msm. overhead\end{tabular}
& \begin{tabular}{@{}c@{}} Depth \& msm. \\ overhead \end{tabular}
\\
\bottomrule
\normalsize{References}
& \cite{glover_logical_2018, lewis_quadratic_2017}
& \cite{harrigan_quantum_2021, otterbach_unsupervised_2017, pelofske_quantum_nodate, shaydulin_qaoa_2023, zhu_multi-round_2022, fujii_deep_2022}
&  {\setlength\extrarowheight{-3pt}\begin{tabular}{@{}r@{}}
{\small circ. cutting:} \cite{bechtold_investigating_2023}\\ 
{\small clustering:} \cite{peng_simulating_2020}\\
{\small causal cones:} \cite{amaro_filtering_2022}
\end{tabular}}
& {\setlength\extrarowheight{-3pt}\begin{tabular}{@{}c@{}}here:~\cite{tan_qubit-efficient_2021}\\ other: \cite{shee_qubit-efficient_2022, glos_space-efficient_2022, fuchs_efficient_2021}
\end{tabular}} 
&  {\setlength\extrarowheight{-3pt}\begin{tabular}{@{}c@{}}\cite{fuller_approximate_2021, teramoto_quantum-relaxation_2023} / \\ \cite{rancic_noisy_2023, winderl_comparative_2022} \end{tabular}}
& \cite{liu_variational_2019}
\\ 
\end{tabular}
\caption{Overview of different methods proposed in the literature to reduce the number of qubits needed in ground state problems. Not all methods are directly applicable to arbitrary QUBO Hamiltonians. Different columns correspond to a rough qualitative classification by the author, grey underlaid what this work is based on. Quantum solvers in row \enquote{QS applicable} are set in braces if they can only be used with some of the references in that column or if further restrictions apply. In the row \enquote{Caveats} we summarize limitations or additional overhead incurred by the methods: evaluation of more circuits, more shots or Pauli measurements, more connectivity-demanding circuits (for shallow decompositions) or deeper circuits. This table focuses on methods that can be used to reduce the number of qubits, applicable in the framework of at least one of the QS in table~\ref{tab:Qsolver_Annealing_QAOA_VQA_assisted}. Qubit-efficient approaches to QUBO problems exist beyond that, e.g.~\cite{dunjko_computational_2018}.
}
\label{tab:qubit-reduction-methods}
\end{table*}

\subsection{Financial Transaction Settlement}\label{subsec:fin_trs_settlement_intro}

We refer to the \emph{transaction settlement problem} as a computational task, consisting of parties $\{1,\ldots,K\}$ with balances $\{\myvec{bal}_k\}$ submitting trades $\{1,\ldots,I\}$ to a clearing house. The task faced by the clearing house is to determine the maximal set of transactions that can be executed without any party $k$ falling below its credit limit $\myvec{lim}_k$. An overview of the notation is given in table~\ref{tab:notation} and a graph representation of a transaction settlement problem with parties as nodes and transactions as edges is shown in figure~\ref{fig:TRS_settlement_example}.

\begin{table}[t]
\centering
\begin{tabular}{@{}l r c l r@{}}
\emph{Financial exchange data} & \emph{Example} & \phantom{a} & \emph{Inputs} & \emph{Math notation} \\ \cmidrule{1-2}\cmidrule{4-5} 
\multicolumn{2}{c}{\cellcolor{gray!10} SETTLEMENT\_INSTRUCTION } & & transactions & $i \in \{1, \ldots, I\}$ \\ 
PARTICIPANT & 205 & & \multirow{2}{*}{parties} & \multirow{2}{*}{$k \in \{1,\ldots, K\}$} \\
COUNTERPARTY & 270 \\
INSTRUMENT & nc157 & & currencies/securities & $j \in \{1, \ldots, J\}$ \\
QUANTITY & 1300 & & transaction value (security) & \multirow{2}{*}{$\myvec{v}_{ik} \in \mathbb{R}^J$} \\
CONSIDERATION & 441.85 & & transaction value (currency) \\
\multicolumn{2}{c}{\emph{transaction weights set to one}} & & transaction weights & $\myvec{w} \in \mathbb{R}^I_{\ge0}$ \\
\multicolumn{2}{c}{\emph{credit limit absorbed in balance}} & & credit limits & $\myvec{lim}_k \in \mathbb{R}^J$ \\
\multicolumn{2}{c}{\emph{balance generated}}  & & balance & $\myvec{bal}_k \in \mathbb{R}^J$ \\
SETTLEMENT\_TYPE & DVP \\
& & & \emph{Decision Variables} & \\ \cmidrule{4-4}
& & & settle transaction or not & $\myvec{x} \in \{0,1\}^I$\\
& & & slack variables & $\myvec{s}_k \in \mathbb{R}^J_{\ge0}$\\
\end{tabular}
\caption{The left part of this table shows the format of the settlement instructions data samples obtained from a regulated financial exchange. Each settlement instruction consists of a sending (PARTICIPANT) and receiving (COUNTERPARTY) party, specifies the security transacted (SECURITY), the quantity traded (QUANTITY) as well as the countervalue (CONSIDERATION) in Singapore dollars. For SETTLEMENT\_TYPE \emph{Delivery Vs Payment} (DVP), a security is traded against a cash settlement. Here, the only alternative is \emph{Free Of Payment} (FOP), in which case only the security is transferred from the seller to the buyer. The right part shows the corresponding problem inputs and their mathematical notation (adapted from~\cite{braine_quantum_2019}).}
\label{tab:notation}
\end{table}

In the case when not all parties have sufficient balances to meet all settlement instructions they are involved with, finding this maximal set can be difficult with classical computing resources. Intuitively, this is because a party's ability to serve outgoing transactions may depend on its incoming transactions, creating many interdependencies between different parties (cf. figure~\ref{fig:TRS_settlement_example}). 
Whilst classical technology is sufficient for current transaction volumes, increases could be expected from more securities in emerging markets and digital tokens, for example. Furthermore, cash shortages make optimization more challenging as it becomes harder to allocate funds optimally among various settlement obligations, determining the priority of different trades and parties and an increased risk of settlement failures. Quantum technologies offer a potential path to mitigate these issues.

Transactions can be conducted both in currencies and securities such as equity and bonds (hence $\myvec{bal}_k$ and $\myvec{lim}_k$ are vector-valued). 
A financial exchange may, for example, handle as many as one million trades involving 500-600 different securities by up to 100 financial institutions (parties) per day. 

\begin{figure}[h!]
    \centering
    \includegraphics[width=\linewidth]{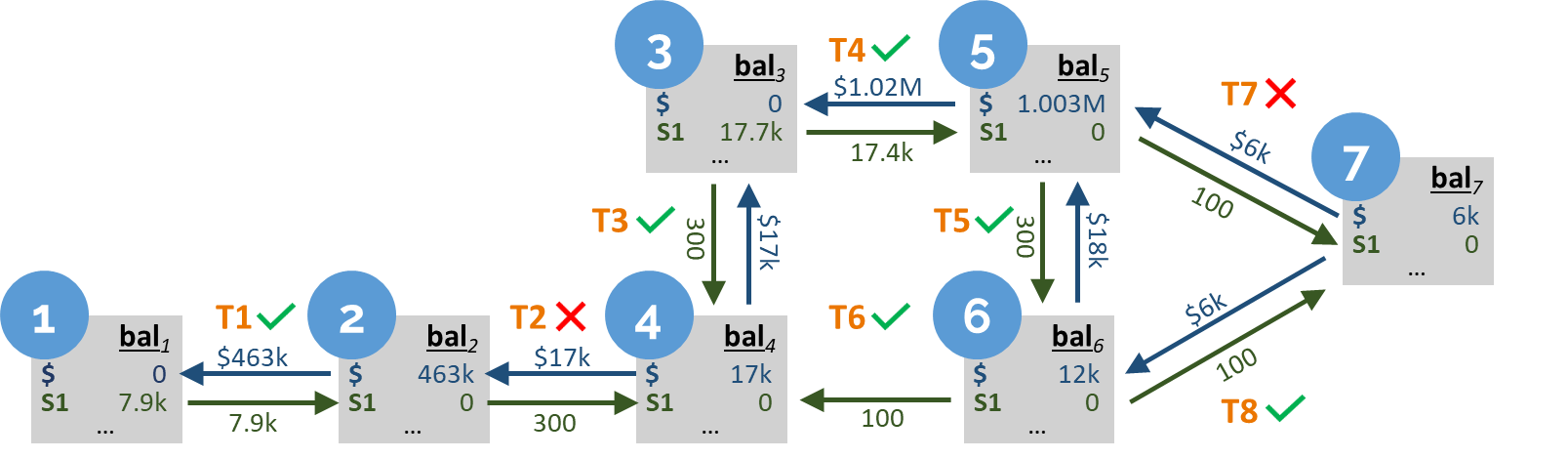}
    \caption{Example of a transaction settlement problem constructed using data provided from a regulated financial exchange. Eight transactions (arrows) between seven parties (numbered squares) are depicted. Each party has initial balances for cash (\$) and different securities (S1, \ldots). The optimal solution which settles the maximal amount of transactions without violating balance constraints is indicated through \textcolor{green}{\checkmark} and \textcolor{red}{$\times$}. The solution is not unique, another optimal solution would settle T2 instead of T3. Even for a problem of only eight transactions, non-trivial dependencies between different transactions exist: For example, T4 can only be settled if T5 is settled which in turn requires T8.}
    \label{fig:TRS_settlement_example}
\end{figure}

\paragraph{QUBO formulation}

To obtain a QUBO formulation of the transaction settlement problem, we follow a slightly simplified version of \cite{braine_quantum_2019}.
The mathematical formulation as a binary optimization problem with inequality constraints looks as follows:
\begin{align}
    \argmax_{\myvec x}~\myvec w^T \myvec x \label{eq:max_wx}\\
    \textrm{subject to \emph{balance constraints}: }\sum_i x_i \myvec v_{ik} + \myvec{bal}_k - \myvec{lim}_k &\geq 0 ~\forall k \in \{1,\ldots, K\}\label{eq:constraints}
\end{align}
where for generality, a weight $w_i$ is given for each transaction $i$ and $\myvec v_{ik}$ represents the balance changes (in cash and securities) for party $k$ in transaction $i$. In practice one might choose $w_i$ proportional to the transaction value of transaction $i$, for simplicity, we will always choose $w_i \equiv 1$.

The solution of this linear constrained binary optimization problem equals the solution of the mixed binary optimization (MBO)
\begin{equation}
    \argmax_{\substack{\myvec x \\ \myvec s_k \geq 0}} \left[ \myvec w^T \myvec x  - \lambda \sum_{k=1}^K \left(\sum_{i=1}^I x_i \myvec v_{ik} + \myvec{bal}_k - \myvec{lim}_k - \myvec s_k\right)^2\right]
    \label{eq:trs_qubo1}
\end{equation}
for large $\lambda$, referred to as the \emph{slack parameter}. Here, continuous \emph{slack variables} $\myvec s_k \geq 0$ (element-wise) were introduced to capture the inequality constraints as penalty terms in the objective. Note, that by approximating $\myvec s_k$ as a binary representation, i.e.~$(\myvec s_k)_i \approx \sum_{l = -L_1}^{l = L_2}\tilde{b}_{kil}2^l, \tilde{b}_{kil}\in \{0,1\}$, the problem could further be transformed into a QUBO problem without any constraints. For large enough $\lambda$, any violation of the constraints~\ref{eq:constraints} will result in a less-than-optimal solution vector.
Equation~\ref{eq:trs_qubo1} can directly be rewritten:
\begin{align}
    &\argmin_{\substack{\myvec x \\ \myvec s_k \geq 0}} \myvec x^T A \myvec x + \myvec b(\myvec s)^T \myvec x + c(\myvec s) \label{eq:trs_settlment_as_qubo_w_slack}\\ 
    \textrm{where: }  A &:= -\lambda V V^T,~\textrm{with }V\in \mathbb{R}^{I\times KJ}, V_{il} := (\myvec v_{ik(l)})_{j(l)} \left[\text{with }k(l) := \lceil \frac{l}{J}\rceil, j(l) := (l\text{ mod }J)+1\right]\label{eq:QUBO_trs_A}\\
    b_i(\myvec s) &:= w_i - 2\lambda \sum_{k=1}^K \left [ \myvec{bal}_k-\myvec{lim}_k-\myvec s_k \right ] \myvec v_{ik}\label{eq:QUBO_trs_b}\\
    c(\myvec s) & := - \lambda \sum_{k=1}^K \left [ \myvec{bal}_k -\myvec{lim}_k - \myvec s_k \right ]^2
\end{align}
For fixed $\myvec s$, the minimization over binary $\myvec x$ is a QUBO problem as in equ.~\ref{eq:QUBO} with $Q = A + \textrm{Diag}[\myvec b(\myvec s)]$\footnote{Here we use: $\myvec x^T \text{Diag}[\myvec b]\myvec x = \sum_{i = 1}^I b_i x_i^2 \overset{x_i^2 = x_i}{=} \myvec b^T\myvec x$}, where $\textrm{Diag}[\myvec b]$ is the matrix with the vector $\myvec b$ on the diagonal and zeros elsewhere.

\subsection{Contribution of this work}

The structure of this work is as follows: In section \ref{subsec:transaction_settlement_methods}, we outline the construction of the transaction settlement instances from data provided by a regulated financial exchange. In section \ref{subsec:qubit_compression}, we use the encoding scheme listed in \cite{tan_qubit-efficient_2021} to reduce the number of qubits required and extend the ideas to include a new variational cost objective and ansatz. 
Section \ref{subsec:simulation_results} presents the results, using the transaction settlement instances generated as a testbed for comparing our methods with QAOA and exploring different encodings. Section \ref{subsec:hw_ionq_ibmq} offers comparisons in the solutions obtained when using the exponential qubit reduction to tackle problems with 128 transactions on real quantum hardware by IonQ and IBM Quantum, exceeding previous results using quantum hardware \cite{braine_quantum_2019}. We present some analysis regarding the results obtained, before concluding with section \ref{04_Conclusion}. 
To the best of our knowledge, this is the first work that tackles mixed binary optimization problems with a qubit-efficient approach on a quantum computer.


\FloatBarrier

\section{Methodology}\label{sec:methodology}

To give an overview of the methodology, we first detail how the financial settlement problem can be constructed using data provided by a regulated financial exchange (\ref{subsec:transaction_settlement_methods}). We will then describe the quantum algorithm (fig~\ref{fig:PQC_overview}) consisting of a heuristic to exponentially reduce the number of qubits (section~\ref{subsec:qubit_compression}), a cost-objective (section~\ref{subsec:cost}), a parameterized quantum circuit to generate solution bit-vectors (section~\ref{subsec:vqa_ansatz}) and finally the classical optimization (section~\ref{subsec:opt}).

\sbox0{\begin{quantikz}
\lstick[wires=2]{Ancilla} 
&\lstick{$\ket{0}$} &\qw &\gategroup[wires=4,steps=2,style={inner sep=6pt}]{} & & &\meter{Z}\\
& \vdots & &  & & &\vdots\\
\lstick[wires=2]{Register}
& \vdots & & 
&  
& &\vdots\\
& \lstick{$\ket{0}$} & \qw & & & &\meter{}
\end{quantikz}}

\begin{figure}[h!]
    \centering
    \begin{subfigure}[t]{0.8\textwidth}
        \begin{tikzpicture}[
stepnode/.style={draw, thick, rounded corners=5pt, minimum width=2cm, minimum height=1cm, align=center},
scale=1
            ]
		
		\node[draw=gray, header = $\text{P{\small arameterized} Q{\small uantum} C{\small ircuit}}$
            ] (pqc) at (0,0) {\usebox0};
		
		\node[rotate=90, font = {\large}] at ($(0.8,-0.1)+(pqc)$) {$U(\myvec \theta^{(n)})$};
		
		\node[stepnode, right=2 cm of pqc, yshift = 0cm] (measstring) {(i) Collect measurements \\ $\myvec m = -1, 1, \ldots 1$};
		
		\node[right=1 cm of measstring] (measms) {$\mathcal{M}=\left\{\myvec m_1, \ldots \myvec m_{n_{\text{shots}}}\right\}$};

            \node[stepnode, below=1cm of measms] (calc_cost) {(ii) Calculate cost function};
  
		\node[below=1cm of calc_cost, font = {\large}] (exp_cost) {$\hat{C}(\myvec \theta) \simeq\mathbb{E}_{\myvec \theta}[C]$};

		\node[stepnode, align=center, left=2cm of exp_cost] (class_opt) {(iii) Classical \\ optimization};

            \node[left=1cm of class_opt, font = {\large}] (theta) {$\myvec \theta^{n+1}$};

		\draw[->] let \p1 = (pqc.east) in (\x1,0) to (measstring.west);
            \draw[->] (measstring) to (measms.west);
            \draw[->] (measms) to (calc_cost);
            \draw[->] (calc_cost) to (exp_cost);
            \draw[->] (exp_cost) to (class_opt);
            \draw[->] (class_opt) to (theta.east);
            \draw[->] (theta.west) -| ($(pqc)+(0.75,-3)$) -- ($(pqc)+(0.75,-2.5)$);
		
	\end{tikzpicture}
        \label{sfig:PQC}
    \end{subfigure}
    \caption{Workflow of a quantum-classical hybrid optimization algorithm. The algorithm involves (i) collecting measurements from the parameterized quantum circuit (PQC) with parameters $\myvec \theta^{(n)}$, (ii) calculating the cost function (and potentially its gradients) from the measurement outcomes, and (iii) optimizing the parameters using classical optimization techniques. The steps are further detailed in sections~\ref{subsec:cost}, \ref{subsec:opt} and \ref{subsec:vqa_ansatz}}
    \label{fig:PQC_overview}
\end{figure}

\subsection{Problem instance - Financial Transaction Settlement}\label{subsec:transaction_settlement_methods}

\paragraph{Dataset} This work uses anonymized transaction data to generate settlement problems of arbitrary size $I$.
The format of the settlement instructions made available for this purpose can be seen in table~\ref{tab:notation}. To generate a problem instance we proceeded as follows:
\begin{enumerate}
    \item Fix the number of transactions ($I$), number of parties ($K$) and an integer $R \leq I$.
    \item Choose $I-R$ random transactions from the dataset and randomly assign them to parties (sender and recipient). A single \emph{transaction} consists of a security being transferred from one party to the other and (if delivery vs payment) a cash transaction in the other direction.\footnote{To avoid confusion, we will refer to \emph{transfers} instead when only considering a single security or cash transfer.}
    \item As our dataset does not provide account balances or credit limits, we set $\myvec{lim}_k = \myvec 0$ and choose minimal non-negative balances $\myvec{bal}_k$ for each party $k$ such that all previously chosen transactions can be jointly executed without any party's balance becoming negative. The balances hence depend on the first $I-R$ transactions chosen. This choice of $\myvec{bal}_k$ is made by considering the net balance-change for each party if all transactions were conducted.
    \item Choose additional $R$ random transactions from the dataset and randomly assign them to parties (without changing the balances assigned in the previous step).
\end{enumerate}
This procedure ensures the optimal solution contains at least $I-R$ valid transactions. Due to the minimal choice of the balances, most of the $R$ transactions chosen last are expected to be invalid in the optimal solution.


To mitigate large differences in transaction volumes between different parties ($\text{S}\$~10-10^6$) as well as different units (cash and different securities) in the data samples, 
we renormalize each party's balance, credit limit and transaction volume:
\begin{align}
    \forall \text{parties } k\in \{1,\ldots, K\}, \forall \text{securities } j\in \{1,\ldots, J\}: \\
    \gamma_{kj} :=& \text{mean}\left(\{|(v_{ik})_j|\}_{i=1}^I\setminus\{0\}\right)\\
    (v_{ik})_j \mapsto& \frac{(v_{ik})_j}{\gamma_{kj}}~\forall i\in \{1,\ldots, I\}\\
    (\myvec{lim}_k)_j \mapsto& \frac{(\myvec{lim}_k)_j}{\gamma_{kj}}\\
    (\myvec{bal}_k)_j \mapsto& \frac{(\myvec{bal}_k)_j}{\gamma_{kj}}
\end{align}
which, using equ~\ref{eq:max_wx} and~\ref{eq:constraints}, does not affect the optimal solution.

We show in lemma~\ref{lemma:connectivity} (appendix~\ref{app:proofs}), that the connectivity for the QUBO matrix of the transaction settlement is bounded by twice the average number of transactions per party, $\frac{4I}{K}$, plus a variance term which vanishes for d-regular graphs.


\subsection{Qubit-efficient Mapping}\label{subsec:qubit_compression}

\paragraph{Mapping QUBO to VQA} 
The underlying idea to solve QUBO problems with VQAs is to use the PQC to generate bit-vectors $\myvec x$. The parameters $\myvec \theta$ of the PQC are then tuned such that the generated $\myvec x$ are likely to approximately minimize equation~\ref{eq:QUBO}. 
Formally:
\begin{equation}\label{eq:QUBO_as_expectation_min}
    \min_{\myvec \theta} \mathbb{E}_{\myvec \theta}[C] := \min_{\myvec \theta} \sum_{\myvec{x} \in \{0,1\}^I} \text{Prob}_{\myvec \theta}(\myvec x) \myvec x^T Q \myvec x = \min_{\myvec \theta}\sum_{i \neq j}p_{ij}(\myvec \theta) Q_{ij} + \sum_i p_i(\myvec \theta) Q_{ii}
\end{equation}
where the right hand side only depends on the marginals $p_{ij}(\myvec \theta) := \text{Prob}_{\myvec \theta}(x_i = 1, x_j = 1)$ and $p_i(\myvec \theta) := \text{Prob}_{\myvec \theta}(x_i = 1)$. 

Instead of searching in a discrete space, this turns the problem into the optimization of the continuous parameters of a generator of bit-vectors. This appears similar but is different from relaxation-based approaches, which often replace binary variables through continuous ones to obtain a more tractable optimization problem whose solutions are projected back to a binary format: Here, the model (the PQC) directly generates bit-vectors $\myvec x$ and if it is expressive enough in the distributions $\text{Prob}_{\myvec \theta}(\myvec x)$ it parameterizes, then an optimal $\myvec \theta$ will generate an optimal bit-vector $\myvec x$ deterministically.

The steps to solve this minimization are shown in figure~\ref{fig:PQC_overview}. 

\smallskip

The core of mapping a QUBO problem to a variational minimization problem therefore consists of specifying how to generate bit-vectors $\myvec x$ with a quantum circuit.
In standard QAOA or variational approaches, this mapping is straightforward (equation~\ref{eq:QUBO_hamiltonian}): As the number of qubits $n_q$ equals the number of variables $I$, we simply measure in the computational basis (Pauli-Z) and associate the outcome 1 (-1) of qubit $q_i$ with the bit $x_i$ equal 0 (1). We use a different mapping, generalizing the qubit-efficient approach in~\cite{tan_qubit-efficient_2021}:
\paragraph{Qubit-efficient binary encoding}

We use $n_a$ qubits (\emph{ancillas}) to represent a subset of $n_a$ bits and $n_r$ qubits (\emph{register}) to provide an address labelling this subset. Compared to the approach for standard QAOA, we (partly) encode the bit-position in a binary encoded number instead of the one-hot encoded qubit-position. Hence the name \enquote{binary encoding} in table~\ref{tab:qubit-reduction-methods}.

Formally, consider a covering $\mathcal{A} = \{A_1, ..., A_{N_r}\},~N_r = 2^{n_r}$ of the set of bit-positions $B = \{1,\ldots,I\}$ with $|A_i| \in \{0, n_a\}~\forall i$ and each $A_i$ ordered. Regard the quantum state $\ket{b_1\ldots b_{n_a}}_{\text{anc}}\otimes \ket{r}_{\text{reg}}$ as corresponding to bit $A_r[l]$ (the $l^\text{th}$ entry of $A_r$) equal to $b_l~\forall l \in \{1,\ldots, n_a\}$. 
This quantum state fixes only the subset $A_r$ of the bits. In general, we interpret
\begin{itemize}
    \item Superpositions in the ancilla state $\leftrightarrow$ probabilistic sampling in the computational basis of different bit-vectors $b_1\ldots b_{n_a}$
    \item Superpositions in the register state $\leftrightarrow$ probabilistic sampling in the computational basis of different bit-sets $A_r$
\end{itemize}
resulting in the general form
\begin{equation}
    \label{eq:quantum_state_qubit_compression}
    \ket{\psi({\myvec \theta})} = \sum_{r=1}^{N_r} \beta_r(\myvec\theta) \left[a_r^{00...0}(\myvec\theta) \ket{00 ... 0}_{\text{anc}} + a_r^{00...1}(\myvec\theta) \ket{00 ... 1}_{\text{anc}} + ... \right]\otimes\ket{r}_{\text{reg}},
\end{equation}
where we already indicated that the PQC parameterized by $\myvec \theta$ determines the values of the register amplitudes ($\beta_r(\myvec\theta)$) and normalized bit-vector amplitudes ($a_r^{b_1\ldots b_{n_a}}$). 
We achieve an exponential compression from $n_q = I$ qubits to $n_q = n_a + \lceil\text{log}_2(I/n_a)\rceil$ in the case of a \emph{disjoint covering} (also \emph{perfect matching}). In general, a covering consisting of $|\mathcal{A}| = R$ bit-sets requires $n_a + \lceil\text{log}_2(R)\rceil$ qubits. 

For the simplest case of the \emph{minimal encoding}, defined by $n_a = 1$, each subset consists of just one binary variable with a total of $n_r = I$ subsets. A quantum state in this encoding can be written as
\begin{equation}
    \label{eq:quantum_state_min_enc}
    \ket{\psi(\myvec \theta)}_{n_a = 1} = \sum_{r=1}^{I} \beta_r(\myvec \theta) \left[a_r^0(\myvec \theta) \ket{0}_{\text{anc}} +  a_r^1(\myvec \theta) \ket{1}_{\text{anc}} \right] \otimes \ket{r}_{\text{anc}}, 
\end{equation}
and represents the bit-vector $\myvec x$ if $|a_r^i| = \delta_{ix_r}$. The total number of qubits required is $n_q = 1 + \lceil\text{log}_2(I)\rceil$.

The large decrease in qubits comes with a few drawbacks:
\begin{itemize}
    \item A single measurement in the computational basis only specifies a subset $A_r$ of the bit-positions, and it is not immediate how to sample full bit-vectors $\myvec x$.
    \item Even arbitrary state-preparation through the PQC may only allow limited distributions on the vector x. Consider for example the minimal encoding: 
    It generates bit-vectors distributed as $\text{Prob}_{\myvec \theta}(\myvec x) = \text{Prob}_{\myvec \theta}^1(x_1)\cdot\ldots\cdot\text{Prob}_{\myvec \theta}^I(x_I) = \prod_r^{I} |a_r^{1}|^2$, where $a_r^{1}$ are the coefficients of the ancilla qubits in equation~\ref{eq:quantum_state_min_enc}, corresponding to a \emph{mean-field} approximation (\cite{veszeli_mean_2021}).
    \item Different from QAOA, the cost objective may no longer correspond to the expectation of a Hermitian observable. This issue and a resolution are discussed in appendix~\ref{app:cost_as_observable}.
\end{itemize}

\paragraph{Sampling algorithm} 

We will adopt a simple \emph{greedy} approach here, which fixes entries of $\myvec x$ as they are sampled throughout multiple measurements and concludes once every entry is sampled. We furthermore determine the covering $\mathcal{A}$ through a k-means-inspired clustering on the graph representation of the problem. 
Given uniform $\beta_r(\myvec \theta)$ and a disjoint covering $\mathcal{A}$, the probability of any one register not being sampled after $n_{\text{shots}}$ measurements is exponentially small, bounded by $\text{exp}(-\frac{n_{\text{shots}}}{N_r})$. 
In practice, we sample multiple bit-vectors $\myvec x$ to find candidates for the optimal solution. This allows us to reduce the average number of measurements by reusing measurement outcomes (in particular those that were sampled multiple times before conclusion of the algorithm). Nonetheless, the qubit compression comes at the cost of significant sampling overhead.


\subsection{Cost function}\label{subsec:cost}

Having specified how to generate bit-vectors from measurement samples of the PQC fully determines $\text{Prob}_{\myvec \theta}(\myvec x)$ and hence the minimization problem in equation~\ref{eq:QUBO_as_expectation_min}. In practice, we cannot access $\text{Prob}_{\myvec \theta}(\myvec x)$ directly, but rather obtain finite-shot measurements on the state prepared by the PQC. Hence, we need to specify an estimator of the expected cost $\mathbb{E}_{\myvec \theta}[C]$. We will refer to this estimator as $\hat{C}(\myvec \theta)$.

For the explicit formulation of $\hat{C}(\myvec \theta)$, we make use of the formulation of $\mathbb{E}_{\myvec \theta}[C]$ in terms of marginal probabilities $p_{ij}(\myvec \theta) = \text{Prob}_{\myvec \theta}(x_i = 1, x_j = 1)$ and $p_i(\myvec \theta) = \text{Prob}_{\myvec \theta}(x_i = 1)$ in equation~\ref{eq:QUBO_as_expectation_min}.
For the latter, we use heuristic estimators $\hat{p}_i(\myvec \theta)$ and $\hat{p}_{ij}(\myvec \theta)$ which are constructed by counting the number of times a certain bit (or pair of bits) was sampled with value equal to one.

The exact formulas for these estimators are given in appendix~\ref{subapp:cost_estimator_formula} with a derivation for disjoint coverings in appendix~\ref{subapp:derivation_disj_covering_estimators}. Intuitively,
\begin{equation}
    \hat{p}_{ij}(\myvec \theta) 
    = (1-\hat{\mu}_{ij})\hat{q}_{ij}+\hat{\mu}_{ij}\hat{p}_i\hat{p}_j
\end{equation}
where $0\leq \hat{\mu}_{ij} \leq 1$ and the asymptotic convergence
\begin{align}
    \hat{p}_i(\myvec \theta) &\xrightarrow{n_{\text{shots}}\to \infty} \frac{\sum_{\substack{r = 1 \\ i \in A_r}}^{N_r} |\beta_r(\myvec\theta)|^2\sum_{\substack{b_k \in \{0,1\}\\b_{l_r(i)}=1}} |a_r^{b_1\ldots b_{n_a}}|^2}{\sum_{\substack{r = 1 \\ i \in A_r}}^{N_r} |\beta_r(\myvec\theta)|^2}  \label{eq:convergence_prob_est_1}\\
    \hat{q}_{ij}(\myvec \theta)&\xrightarrow{n_{\text{shots}}\to \infty} \frac{\sum_{\substack{r = 1 \\ i,j \in A_r}}^{N_r} |\beta_r(\myvec\theta)|^2\sum_{\substack{b_k \in \{0,1\}\\b_{l_r(i)}=b_{l_r(j)}=1}} |a_r^{b_1\ldots b_{n_a}}|^2}{\sum_{\substack{r = 1 \\ i,j \in A_r}}^{N_r} |\beta_r(\myvec\theta)|^2} \label{eq:convergence_prob_est_2}
\end{align}
motivates the expressions $\hat{p}_i$ and $\hat{p}_{ij}$.

Following equation~\ref{eq:QUBO_as_expectation_min}, the cost estimator $\hat{C}(\myvec \theta)$ for the transaction settlement problem then takes the form 

\begin{equation}
    \hat{C}(\myvec \theta) = \sum_{\substack{i,j = 1\\i\neq j}}^I \hat{p}_{ij}(\myvec \theta) A_{ij} + \sum_{i = 1}^I \hat{p}_{i}(\myvec \theta) (A_{ii}+b_i(\myvec s)) + c(\myvec s)
    \label{eq:cost_estimator}
\end{equation}

which is optimized with respect to $\myvec \theta$ and the slack variables $\myvec s = (\myvec s_1, \ldots, \myvec s_K) \geq 0$. 
The optimal slack variables can be obtained straightforwardly using 
\begin{equation}
    \hat{\myvec s}_k(\myvec \theta) = \max \left(\myvec 0, -\myvec{lim}_k + \sum_{i=1}^I \hat{p}_i(\myvec \theta) \myvec v_{ik}+\myvec{bal}_k\right)
    \label{eq:slack_estimator}
\end{equation}
where $\max(\circ, \circ)$ is to be taken element-wise. 
The optimal slack variables substitute $\myvec s$ in equ~\ref{eq:cost_estimator}, thus removing the need for separate optimization over the slack variables. 

\paragraph{Remarks:} 
\begin{enumerate}
    \item It is not possible to express equation~\ref{eq:cost_estimator} as the expectation of a Hermitian observable on a state of the form of equation~\ref{eq:quantum_state_qubit_compression} due to denominators in the expressions for $\hat{p}$, $\hat{q}$ (equ.~\ref{eq:p_hat_pauli}, \ref{eq:q_hat_pauli} in appendix) and $\hat{\mu}_{ij}$ (equ.~\ref{eq:estimator_mu}) as well as the functional form of $\hat{\myvec s}$. We will show in appendix~\ref{app:cost_as_observable} how this problem can be resolved for fixed $\myvec s$ given uniform $\beta_r(\myvec \theta)$.
    \item In the limit of the \emph{full encoding} ($n_a = I$, $n_r = 0$), we get \[\hat{p}_{ij} = \hat{q}_{ij} = \left \langle \mathds{1}_{\min\{i,j\}-1}\otimes\ket{1}\bra{1}\otimes\mathds{1}_{|i-j|-1}\otimes\ket{1}\bra{1}\otimes \mathds{1}_{I-i-j}  \right \rangle_{\mathcal{M}} \]
    and
    \[\hat{p}_i = \left \langle \mathds{1}_{i-1}\otimes\ket{1}\bra{1}\otimes \mathds{1}_{I-i}  \right \rangle_{\mathcal{M}}\]
    resulting in the \enquote{standard} cost estimator identical to e.g.~QAOA.
    \item Runtime and memory cost: 
    The naive run-time for classically computing the cost estimator scales as $O(n_{\text{shots}}I^2)$, while the memory required only scales as $O(n_{\text{shots}}+I^2)$. Our approach does not require full tomography with memory requirements as high as $O(4^I)$. 
    In further extensions, methods such as classical shadows (\cite{huang_predicting_2020}) may be used to more efficiently estimate the cost and reduce $n_{\text{shots}}$.
\end{enumerate}

If certain registers are hardly sampled, i.e. $ |\beta_r(\myvec\theta)|^2\simeq 0$, we may encounter division by zero in the expressions for $\hat{p}_i$ and $\hat{q}_{ij}$. In practice, this can be dealt with by setting the corresponding estimators to $1/2$ whenever estimates for $|\beta_r(\myvec\theta)|^2$ fall below some $\epsilon > 0$, resulting in indirect penalization. Alternatively, we can add an explicit regularization term $\hat{R}(\myvec \theta) = \eta \sum_{r = 1}^{N_r}\left [\hat{r}_r(\myvec \theta) -\frac{1}{N_r} \right]^2$ to the cost function. 

\subsection{Variational ansatz}\label{subsec:vqa_ansatz}

In this work, we consider two types of PQC:
\begin{itemize}
    \item A \emph{hardware-efficient} ansatz consisting of RY rotations and entangling CNOT layers.
    \item A \emph{register-preserving} ansatz of conditional RY rotations incorporating constraints and symmetries tailored to the qubit-efficient encoding.
\end{itemize}
Both variational circuits are depicted in figure~\ref{fig:var_ansatz}. The hardware-efficient ansatz was used identically in~\cite{tan_qubit-efficient_2021}, the register-preserving ansatz is one of the main contributions of this work.

We mentioned difficulties arising from vanishing register-amplitudes in the previous section. We will now formally define register-uniform quantum states and register-preserving circuits before discussing the advantages offered by them:

\begin{restatable}[Register-uniform]{de}{definitionreguniform}
    \label{def:reg_uniform}
    We call a quantum state $\ket{\psi}_{ar}\in \mathcal{H}_{\text{anc}}\otimes\mathcal{H}_{\text{reg}}$ \emph{register-uniform} with respect to the orthonormal basis $\left(\ket{1},\ldots ,\ket{N_r}\right)$ of $\mathcal{H}_{\text{reg}}$, if it can be written as
    \begin{equation}
        \ket{\psi} = \frac{1}{\sqrt{N_r}}\sum_{r = 1}^{N_r}\ket{\phi_r}_{\text{anc}}\otimes\ket{r}_{\text{reg}}
        \label{eq:reguniform_def}
    \end{equation}
    where $\ket{\phi_r}_{\text{anc}} \in \mathcal{H}_{\text{anc}}$ is arbitrary with $\braket{\phi_r | \phi_r}_{\text{anc}} = 1 ~\forall r$. 
\end{restatable}

\begin{restatable}[Register-preserving]{de}{definitionregpreserving}
    \label{def:reg_preserving}
    We call a unitary $U$ acting on $\mathcal{H}_{\text{anc}}\otimes\mathcal{H}_{\text{reg}}$ \emph{register-preserving} with respect to the orthonormal basis $\left(\ket{1},\ldots ,\ket{N_r}\right)$ of $\mathcal{H}_{\text{reg}}$ if it always maps register-uniform states to register-uniform states (with respect to the same basis).
\end{restatable}
The set of register-preserving unitaries with respect to the same basis is closed under concatenation. 
Our register-preserving ansatz first prepares the register-uniform plus state $\ket{+}^{\otimes n_q} = H^{\otimes n_q} \ket{0}$ ($H$ being the Hadamard gate) and then acts through register-preserving unitaries on it.

\paragraph{Notation} We use the bra-/ket-notation only for normalized states. Furthermore, Latin letters inside bra and ket indicate computational basis states, while Greek letters indicate general quantum states.
We will refer to \emph{register-preserving circuits} as quantum circuits which output register-uniform states with respect to the computational basis (and fix said basis from now on, omitting further mention of it).

The following claims about register-uniform states and register-preserving circuits are proved in appendix~\ref{app:proofs}:

\begin{restatable}[]{lm}{lemmareguniform}
    \label{lemma:reg_uniform}
    The following is equivalent to a state $\ket{\psi}\in \mathcal{H}_{\text{anc}}\otimes\mathcal{H}_{\text{reg}}$ being register-uniform:
    \begin{equation}
        \bra{r}\text{tr}_{\mathcal{H}_{\text{anc}}}(\ket{\psi}\bra{\psi})\ket{r} = \frac{1}{N_r} ~\forall r\in \{1,\ldots, N_r\}
        \label{eq:reguniform_equivalent_def}
    \end{equation}
\end{restatable}

\begin{restatable}[]{thm}{theoremregpreserving}
    \label{theorem:reg_preserving}
    The following are equivalent for a unitary $U$ acting on $\mathcal{H}_{\text{anc}}\otimes\mathcal{H}_{\text{reg}}$:
    \begin{enumerate}[label=(\roman*)]
        \item $U$ is register-preserving
        \item $U (\ket{\phi}_{\text{anc}}\otimes\ket{r}_{\text{reg}}) = (U_r\ket{\phi}_{\text{anc}})\otimes\ket{f(r)}$ where $U_r$ is a unitary on $\mathcal{H}_{\text{anc}}$ $\forall r$ and $f: \{1,\ldots,N_r\}\to\{1,\ldots,N_r\}$ is bijective.
        \item $U$ can be written as a sequence of unitary matrices on $\mathcal{H}_{\text{anc}}$ conditioned on a subset of register-qubits and basis-permutations on the register.
    \end{enumerate}
\end{restatable}
Note\footnote{To avoid confusion with the \emph{permutation operator} used in quantum physics which refers to permuting particle-labels.}, that in theorem~\ref{theorem:reg_preserving} we refer to \emph{basis}-permutations, not \emph{qubit}-permutations, although former contains the latter.

\begin{figure}[h]
    \centering
    \begin{subfigure}[t]{0.7\textwidth}
    \resizebox{\textwidth}{!}{
        \begin{quantikz}
        \lstick[wires=2]{Ancilla} 
            & \gate[wires=1]{R_Y(\theta_1)} \gategroup[wires=5,steps=5,style={dotted, cap=round, inner sep=7pt}]{Conditional rotations} &\qw ~...~ &\gate[wires=1]{R_y(\theta_{n_r})} & \qw &\qw &\qw &\qw \gategroup[wires=5,steps=3,style={dotted, cap=round, inner sep=7pt}]{Basis permutation} & \qw & \qw \\
            &\qw   &\qw &\qw &\qw & \gate[wires=1]{R_y(\theta_{n_r+1})} &\qw ~...~ & \qw & \qw&\qw\\
        \lstick[wires=3]{Register} 
            & \ctrl{-2} &\qw &\qw&\qw &\ctrl{-1} & \qw &\qw &\qw &\targ{}\\
            ~~~\raisebox{0.2cm}{\vdots} & & \qw ~...~  & \qw & \qw & \qw & \qw ~...~ & \targ{} & \qw ~...~ &\ctrl{-1}\\
            & \qw& \qw & \ctrl{-4} &\qw & \qw & \qw ~...~ & \ctrl{-1} & \qw &\qw
        \end{quantikz} 
    }
    \caption{Register-preserving ansatz}
    \label{fig:reg-pres_ansatz}
    \end{subfigure}
    \begin{subfigure}[t]{0.27\textwidth}
    \resizebox{\textwidth}{!}{
        \begin{quantikz}
        \lstick[wires=2]{Ancilla} 
            & \gate[wires=1]{R_Y(\theta_{1~})} \gategroup[wires=5,steps=1,style={dotted, cap=round, inner sep=7pt}]{Rotations} & \qw & \gategroup[wires=5,steps=2,style={dotted, cap=round, inner sep=7pt}]{Entangling} \qw & \targ{}\\
            &\gate[wires=1, ]{R_y(\theta_{n_a})} &\qw & \targ{} & \ctrl{-1}\\
        \lstick[wires=3]{Register} 
            &  ~~~\raisebox{0.2cm}{\vdots}~~~ & \qw & \ctrl{-1} & \targ{}\\
            & ~~~\raisebox{0.2cm}{\vdots}~~~ & \qw & \targ{} & \ctrl{-1}\\
            & \gate[wires=1]{R_y(\theta_{n_q})} & \qw & \ctrl{-1} & \qw
        \end{quantikz} 
    }
    \caption{Hardware-efficient ansatz}
    \label{fig:hwe_ansatz}
    \end{subfigure}
    \caption[Variational Ansatz]{One layer of the register-preserving ansatz (left) and the hardware-efficient ansatz (right) used in the result section. To obtain the PQC $U(\myvec \theta)$, these circuit layers $L$ are preceded by a Hadamard gate on every qubit and are repeated a chosen number $d$ of times (\emph{depth}), i.e. $U(\myvec \theta) = L(\myvec \theta)^d H^{\otimes n_q}$. Both ansätze contain parameterized rotations followed by a layer of CNOTs. Here, the basis permutation layer of the register-preserving ansatz is identical to the entangling layer of the hardware-efficient ansatz, with the difference that it only acts on register-qubits. Furthermore, in the case of the register-preserving ansatz, parameterized single-qubit $R_Y$ rotations were applied to all ancillas prior to the first layer.}
    \label{fig:var_ansatz}
\end{figure}

Theorem~\ref{theorem:reg_preserving} provides a list of ingredients that may be used to construct register-preserving variational ansätze. Namely, we can combine conditional unitaries (such as CNOT, Toffoli gates), acting on the ancillas and conditioned on the register-qubits, with arbitrary unitaries that only act on the ancilla qubits. Furthermore, we can permute computational basis states on the register qubits. These permutations could be cryptographic permutation pads~\cite{kuang_quantum_2022}, binary adder circuits~\cite{draper_addition_2000} or heuristic constructions from NISQ-friendly gates such as CNOT, SWAP and X gates.

When defining register-preserving circuits, we demand that \emph{any} register-uniform state is mapped to a register-uniform state. This may not always be necessary. In the case of this work, we always start with the same input state $\ket{+}^{\otimes n_q}$ which allows more general unitaries than theorem~\ref{theorem:reg_preserving}, as the following lemma demonstrates:

\begin{restatable}[]{lm}{lemmaregpreservingsometimes}
    \label{lemma:reg_pres_sometimes}
    For register-uniform states as in equation~\ref{eq:reguniform_def} with $\braket{\phi_{r_1}|\phi_{r_2}}_{\text{anc}}\in \mathbb{R}~\forall r_1, r_2\in \{1,\ldots, N_r\}$, a unitary $U = \mathds{1}_{\text{anc}}\otimes U_\text{reg}$ only non-trivially acting on the register-qubits always maps $\ket{\psi}$ to a register-uniform state if and only if
    \begin{equation}
        (U^\dagger)_{r_1s}U_{sr_2} \in i\mathbb{R}~ \forall r_1,r_2,s\in \{1,\ldots, N_r\} \text{ with } r_1<r_2.
    \end{equation}
\end{restatable}
While theorem~\ref{theorem:reg_preserving} only allows permutations on the register-qubits, this lemma allows (a single) application of $\text{exp}\left(-i\frac{\theta}{2}P\right)$ for any self-inverse permutation $P$, e.g. R$X^n$ rotations or the RBS gate, on the register-qubits.
The condition $\braket{\phi_{r_1}|\phi_{r_2}}_{\text{anc}}\in \mathbb{R}$ is trivially fulfilled for states which are a real-valued linear combination of computational basis states such as $\ket{+}^{\otimes n_q}$.


\paragraph{Our ansatz} The circuits used in this work are shown in figure~\ref{fig:var_ansatz}. The register-preserving circuit acts with conditional RY rotations on every ancilla-qubit, conditioned on individual register-qubits. The RY rotation on ancilla-qubit $b$ conditioned on register-qubit $c$ can be regarded as a parameterized rotation on $b$ for half the registers (those registers $r$, for which the binary encoding of r has a 1 at position $c$). A basis permutation layer consisting of CNOTs is added to ensure consecutive conditional RY rotations act on a different set of registers (this basis permutation layer is omitted if only a single layer is used, $d=1$).

In terms of optimization parameters, we optimize $n_a*n_r$ parameters per register-preserving layer and $n_q = n_r+n_a$ for the hardware-efficient ansatz.

\paragraph{Discussion register-preserving ansatz}
Only allowing register-preserving gates in the variational ansatz imposes challenges in keeping the variational ansatz both expressive and NISQ-friendly, at least on superconducting hardware (cf. section~\ref{subsec:hw_ionq_ibmq}). On the other hand, we see the following motivations and advantages for exploring register-preserving circuits:
\begin{enumerate}
    \item Respect the symmetries of the qubit-efficient approach: 
    In light of challenges associated with barren plateaus for over-expressive ansätze, incorporating symmetries into the circuit architecture -- here: register-preservation and real-valued amplitudes\footnote{Real-valued amplitudes are the reason we only make use of RY rotations (instead of RX, RZ).} in the computational basis -- is promising as it has been shown to help with the problem of vanishing gradients (\cite{schatzki_theoretical_2022, liu_variational_2019}).
    \item Numeric stability: The cost estimator $\Hat{C}(\myvec \theta)$ (equ.~\ref{eq:cost_estimator}) makes use of estimators for the register-amplitudes $|\beta(\myvec \theta)|^2$. These can be fixed to $\frac{1}{N_r}$ for a register-preserving circuit, adding numerical stability (especially as the terms affected are in the denominator) and reducing the computational overhead. In figure~\ref{fig:shot_noise_variance_grad}, we visualize the variance of the gradient-estimator with respect to shot noise at fixed parameters $\myvec \theta$. The register-preserving ansatz shows much smaller variance, which suggests a lower number of required shots $n_{\text{shots}}$ (cf.~bullet 3.). 
    \item Sampling overhead: Register-uniform states minimize the expected number of samples needed to cover each register (\cite{boneh_coupon-collector_1989}).  
    Furthermore, theorem~\ref{theorem:reg_preserving} $(ii)$ shows, that the net effect of any register-preserving unitary on the register-qubits is a permutation. If this permutation $P: \ket{r}\mapsto \ket{f(r)}$ is easily inverted, then the bit-vector sampling can be made deterministic in the register (without otherwise impacting the prediction), by using the input state $H^{\otimes n_a}\ket{0}_{\text{anc}}\otimes P^{-1} \ket{r}_{\text{reg}}$ instead of $H^{\otimes n_a+n_r}\ket{0}$. In any case, we can reduce the number of circuit evaluations needed by using initial states of the form $H^{\otimes n_a}\ket{0}_{\text{anc}}\otimes \ket{r}_{\text{reg}}$ and ensuring that we sample a different register $f(r)$ in every run by iterating over $r \in \{1,\ldots, N_r\}$. 
    \item Expression as expectation value of Hermitian observable: As all denominators in the expressions for $\hat{p}_i$, $\hat{q}_{ij}$ and all of $\hat{\mu}$ are replaced by constants, this allows -- for fixed $\myvec s$ -- to express $\hat{C}(\myvec \theta)$ as the expectation value of a Hermitian observable (although the product $\hat{p}_i \hat{p}_j$ requires preparation of a product state, see appendix~\ref{app:cost_as_observable}).
    A majority of the literature (including aforementioned classical shadows) and software are tailored primarily for Hermitian expectation values. 
    Areas include theoretical results (e.g.~adiabatic theorem), the variational ansatz and optimizer itself, estimation and error mitigation as well as fault-tolerant methods for the evaluation of expectation values.
    Expressing our cost function as a Hermitian expectation hence widens the cross-applicability of other results and code-bases.
\end{enumerate}

\begin{figure}
    \captionsetup{width = 0.5\linewidth}
    \floatbox[{\capbeside\thisfloatsetup{capbesideposition={right,top},capbesidewidth=0.5\linewidth}}]{figure}
    {\caption[Shot-noise gradient]{\colorbox{lightgray!30!white}{16 Transactions, 12 parties, 6 qubits, $n_a = 4$} | Variance in partial derivatives of cost estimator over ten samples ($10^4$ shots each). The variance was in turn averaged over all entries of the gradient. This was done for circuits of different depths (x-axis), for each of which we uniformly sampled 25 different parameters $\myvec \theta$ (scatter). The median over the different parameters is depicted as a solid line, the inter-quartile range as a shaded region.}\label{fig:shot_noise_variance_grad}}
    {\includegraphics[width=1\linewidth]{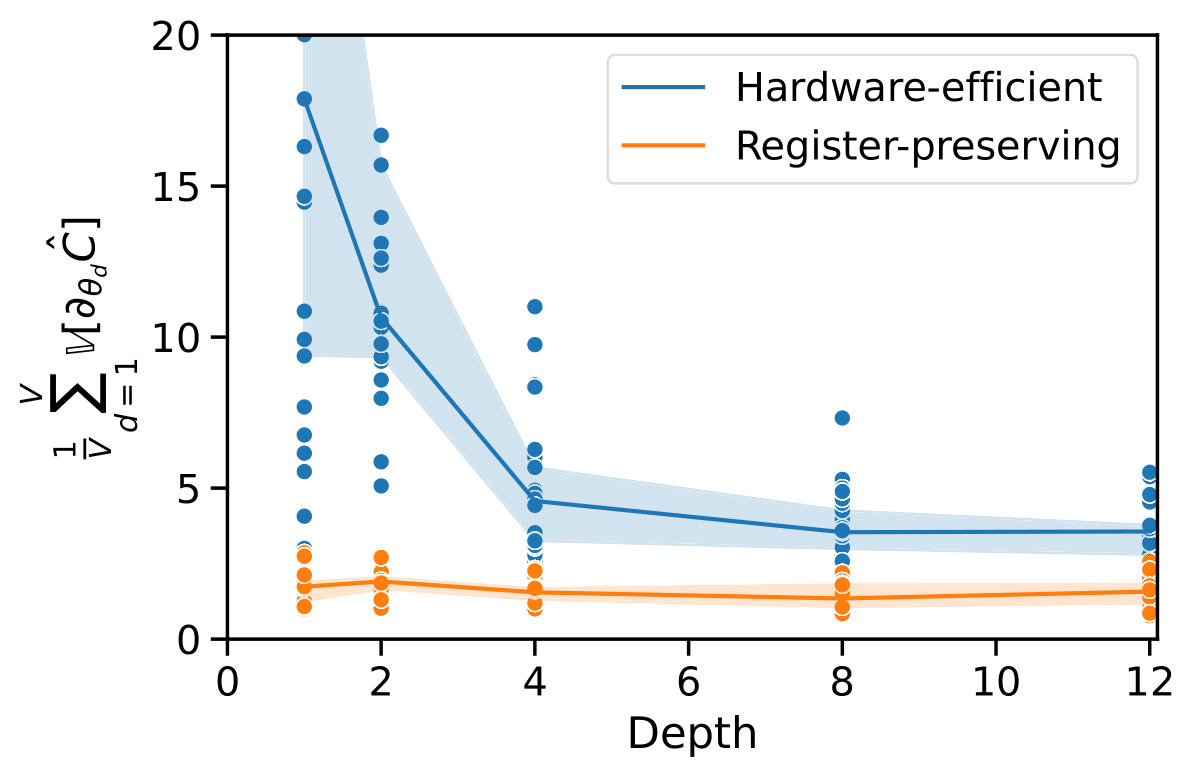}}
\end{figure}

\subsection{Optimization}\label{subsec:opt}

Many different optimization procedures have been suggested in the literature to find the optimal parameters for a PQC through classical optimization. This includes the parameter-initialization (\cite{truger_warm-starting_2023, egger_warm-starting_2021, akshay_parameter_2021, grant_initialization_2019, mitarai_generalization_2019}), choice of meta-parameters (\cite{goh_techniques_2022}) as well as the parameter update itself (\cite{barkoutsos_improving_2020, schuld_evaluating_2019, ostaszewski_structure_2021, nakanishi_sequential_2020, skolik_layerwise_2021}, an overview of gradient-based and gradient-free optimizers can be found in section D. of \cite{bharti_noisy_2022}). 

While the optimization of a PQC has been shown to be NP-hard (\cite{bittel_training_2021}) and may well be the most important ingredient to practical advantage for any quantum QUBO solver, the focus of this work is on the qubit-efficient methods rather than on the optimization itself. Our results were obtained with two different commonly used optimizers: The gradient-free optimizer COBYLA (implemented in scipy~\cite{virtanen_scipy_2020}) as well as standard gradient descent, with gradients calculated through the parameter-shift rule (\cite{mitarai_quantum_2018, schuld_evaluating_2019}). 

The full optimization step for updating the circuit parameters $\myvec \theta^{(n)}\mapsto \myvec \theta^{(n+1)}$ is depicted in figure~\ref{fig:PQC_overview}.

\FloatBarrier
\begin{figure}
    \centering
    \includegraphics[width = 0.95\linewidth]{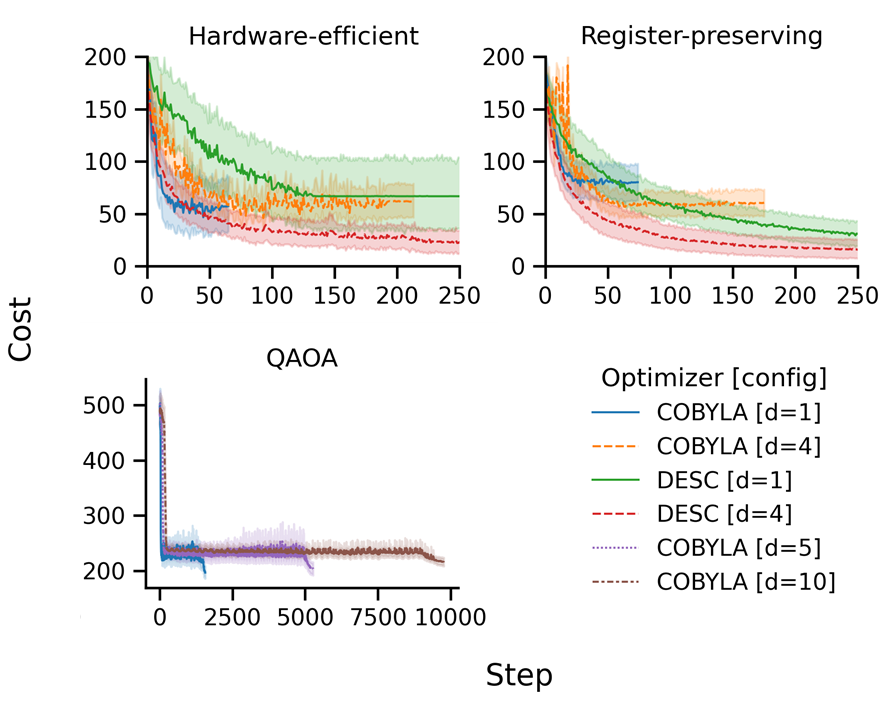}
    \caption[Training traces]{\colorbox{lightgray!30!white}{16 Transactions, 10 parties, 5 (register-efficient) / 16 (QAOA) qubits, $n_a = 1$, $2\times 10^4$ shots} | Cost during parameter optimization for different circuit ansätze (Hardware-efficient, Register-preserving and QAOA) applied to a transaction settlement problem with 16 transactions. Two different optimizers, COBYLA (gradient-free) and gradient descent (DESC) were used in the qubit-efficient approach while only COBYLA was used for QAOA. The depth of the circuit (in number of layers/p-value for QAOA) is given in square brackets. The solid lines depict means over 25 (qubit-efficient) and 10 (QAOA) different training runs with random starting points. The 95\% confidence interval is given as a shaded region. Quantitative values against QAOA are not comparable due to different cost functions used for the parameter training.} 
    \label{fig:training_trace_16TRS}
\end{figure}

\section{Results and Discussion}\label{sec:results}


Here, we present results from applying the methodology presented in section~\ref{sec:methodology} to transaction settlement problems of 16 and 128 transactions. We compare hardware-efficient and register-preserving qubit compression with QAOA. We show results for both a simulator backend (Pennylane \cite{bergholm_pennylane_2022}) and quantum hardware from IBM Quantum and IonQ. The statistics for uniformly random solution-sampling is also provided for benchmarking. For 16 transactions, this includes the optimal solution. 
Throughout this section $R$ (see~\ref{subsec:transaction_settlement_methods}) was set to $\lfloor \frac{I}{4}\rfloor$ and we considered only cash and one security ($J=2$). We randomly generated three sets of $I=$ 16 transaction instructions with $K=$ 10, 12 and 13 parties respectively and one settlement problem with 128 transaction instructions, $K = 41$.


\subsection{Simulation, 16 transactions -- Comparison with QAOA}\label{subsec:simulation_results}

\paragraph{Training convergence} Figure~\ref{fig:training_trace_16TRS} shows the training convergence during the parameter optimization, averaged over different random initial parameters of the PQC. While COBYLA returned optimized parameters within a few hundred steps or less, its cost value is consistently outperformed by gradient descent (DESC), especially for an increasing number of circuit parameters.

The register-preserving ansatz not only outperforms the hardware-efficient PQC, but also produces solutions with less variance for different starting points. For all the qubit-efficient approaches, deeper circuits also improved the performance.

For QAOA, the substitution in equ~\ref{eq:slack_estimator} to optimize both slack variables and variational parameters simultaneously is infeasible as the variational ansatz depends on the QUBO matrix and by extension, the slack variables (cf. equation~\ref{eq:trs_settlment_as_qubo_w_slack}). Results for QAOA were obtained by alternating the optimization of slack variables and circuit parameters 50 times, with up to $10^3$ COBYLA-iterations to update the circuit parameters at each cycle. The optimization landscape appears to be dominated by the slack variables, and each update changes the optimization landscape for the variational parameters. This unusual optimization landscape is likely the reason why no significant improvements were observed for increasing p-values.

From our brief comparison, hardware-efficient ansätze appear to be more suited for MBO problems as they are agnostic to changes in the QUBO matrix. Despite these challenges, we maintain our QAOA results for the purposes of comparison and leave the exploration of more effective implementations of the QAOA to MBO to future work.

\paragraph{Bit-vector quality} As the cost estimator used in the optimization is only a proxy for the actual quality of the bit-vectors $\myvec x$ generated, we show the empirical cumulative distribution of the cost associated with bit-vectors generated from the trained PQC in figure~\ref{fig:cdf}. We normalized the cost for each transaction settlement problem and averaged over different configurations (three settlement problems, $n_a \in \{1,4,8\}$ for~\ref{fig:cdf_QAOA}, COBYLA and gradient descent, up to 25 training runs), drawing 50 bit-vectors per configuration.

Subfigure~\ref{fig:cdf_QAOA} shows the cumulative distributions for both the qubit-efficient approach and QAOA.  Except for the hardware-efficient ansatz with one layer, our qubit-efficient approach performs better than QAOA on average. As in the training traces, no significant differences in the results for QAOA were found by varying the depth (p-value) from one to ten. The register-preserving ansatz performs best for all depths.

In subfigure~\ref{fig:cdf_AncReg}, weak improvement can be observed by using 8 instead of 1 ancilla qubits and by adding another register-qubit ($n_r+ =1$). 

During the training, we observed that gradient descent yields better minima than COBYLA in the cost estimator but tends to \emph{sparse} solutions, i.e. the associated distribution on bit-vectors is strongly concentrated around a single value (cf. figure~\ref{fig:cdf_HW}). Here, redundant encoding of bit-vector-positions in the ancillas (i.e.~$n_r+ > 0$) was found to help in generating more diverse solution candidates.

\begin{figure}
    \centering
    \begin{subfigure}[t]{0.49\textwidth}
        \includegraphics[width =1.0\linewidth]{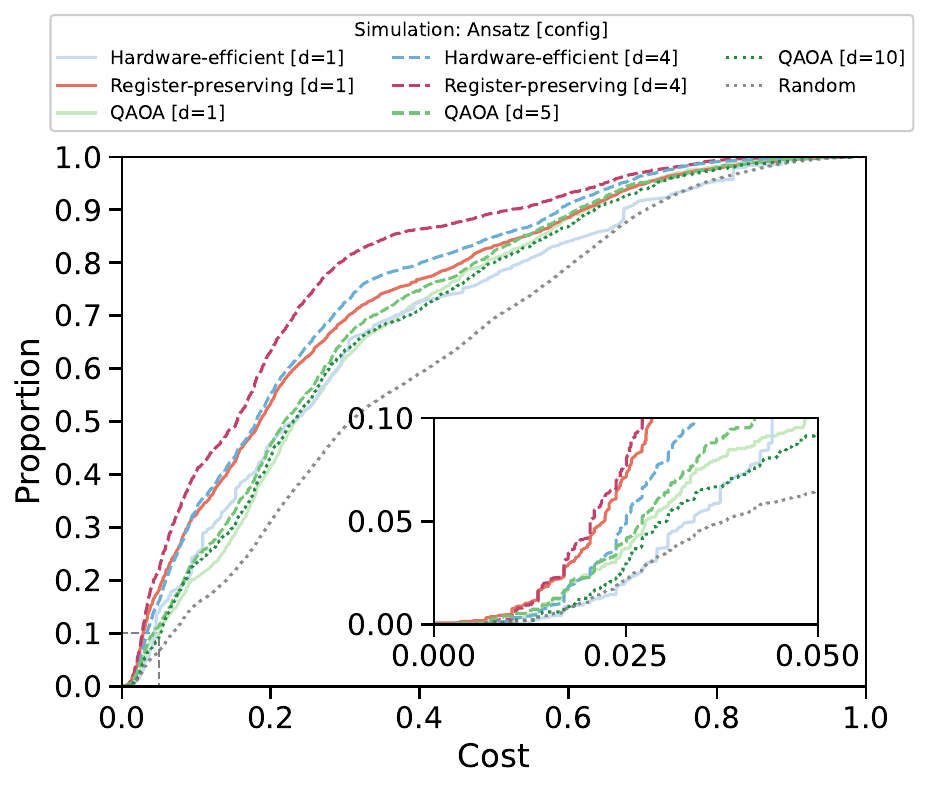}
        \caption{Different circuit ansätze (hardware-efficient, register-preserving and QAOA) are compared at different depths (number of layers/p-value). $n_a, n_r+ \in \{(1,0),(4,0),(8,0),(8,1)\}$, COBYLA.}
        \label{fig:cdf_QAOA}
    \end{subfigure}
    \begin{subfigure}[t]{0.475\textwidth}
        \includegraphics[width =1.0\linewidth]{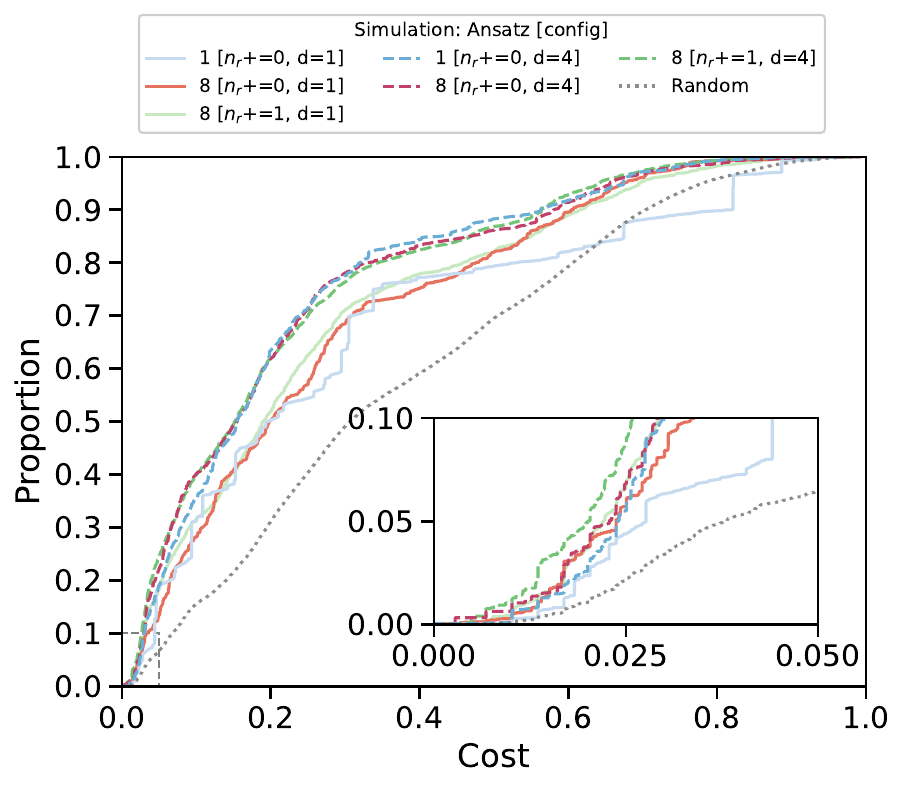}
        \caption{Different encodings with one and eight (both as a disjoint covering, $n_r+ = 0$, and as a redundant covering, $n_r+ = 1$) ancilla qubits are compared for different optimization algorithms. Both COBYLA and gradient descent are included.}
        \label{fig:cdf_AncReg}
    \end{subfigure}
    \caption[Empirical cumulative distribution function]{\colorbox{lightgray!30!white}{16 Transactions, 10-13 parties, 5-16 qubits, $10^4$ shots} | Empirical cumulative distribution function of normalized cost over bit-vectors generated by different trained PQCs: For a given cost x, the empirical probability y of generating bit-vectors with cost at most x is plotted. Due to a rescaling of the x-axis (cost normalization), an x-value of zero corresponds to the optimal bit-vector and a value of one to the bit-vector maximizing the cost. Both can be easily found for a problem with only 16 transactions. Insets show the area with the lowest 5\% of the cost. In the calculation of the cost, equation~\ref{eq:trs_settlment_as_qubo_w_slack} was used with optimal slack-variables $\myvec s (\myvec x)$ for a given bit-vector $\myvec x$. The distribution for uniformly random bit-vectors is shown as a grey dotted line.}
    \label{fig:cdf}
\end{figure}


\subsection{Hardware, 16 \& 128 transactions -- Results on IonQ and IBMQ QPUs}\label{subsec:hw_ionq_ibmq}

To investigate the generation of bit-vectors on real quantum hardware (QPU), we optimized different configurations of both register-preserving and hardware-efficient PQCs on a simulator for 16 and 128 transactions. The pre-trained circuit parameters were then executed on the Geneva/Hanoi QPU provided by IBM Quantum and the Harmony/Aria QPU provided by IonQ.\footnote{Different backends from both providers were used as \texttt{ibmq\_geneva} was retired while this work was in process and \texttt{ionq\_harmony} only provides 11 qubits, necessitating the larger \texttt{ionq\_aria} device for 128 transactions with 16 ancillas.}
The resulting cost-distributions of generated bit-vectors for a settlement problem with 16 transactions and 10 parties / 128 transactions and 41 parties are depicted in figure~\ref{fig:cdf_HW} and \ref{fig:ionq_ibm_128trs_HW}.




\paragraph{IBMQ vs IonQ} For $4$ layers of the register-preserving circuit (\ref{fig:cdf_reg-pres-HW}), the IBMQ results are significantly worse than for IonQ. This can partly be attributed to the connectivity requirements of the long-range conditional Y-Rotations used in the register-preserving ansatz (figure~\ref{fig:reg-pres_ansatz}). This favours the all-to-all connectivity of \texttt{ionq\_harmony}, which foregoes the need for depth-increasing SWAP networks. The hardware-efficient ansatz (\ref{fig:cdf_hwe-HW}) on the other hand is compatible with the lattice connectivity of IBMQ devices and shows similar performance for both QPUs. 


For 128 transactions, the benchmarked IonQ device (\texttt{ionq\_aria}) slightly outperforms IBMQ (\texttt{ibm\_hanoi}) even with the hardware-efficient ansatz, the results of which are depicted in figure~\ref{fig:ionq_ibm_128trs_HW}.


\paragraph{Impact of noise}
In general, the noisy results obtained from real quantum backends yielded worse results than noise-free simulations. However, the additional variation in the generated bit-vectors could also help to generate solutions of lower cost. This is observed in figure~\ref{fig:cdf_HW_128TRS}, where the hardware results yielded bit-vectors of lower cost than the lowest simulated vectors with a probability of more than 10\%. Noise does not necessarily move the distribution towards uniform random bit-vectors: On real hardware, the decay into the physical ground state is more likely than the excited state. Depending on the $\sigma_z$-to-bit mapping, this can result in a larger or smaller number of settled transactions than uniform randomness and, potentially, in performance worse than uniform random (as for \texttt{ibm\_geneva} in figure~\ref{fig:cdf_reg-pres-HW}).

\smallskip

Figure~\ref{fig:ViolxTRS_HW_128TRS} emphasizes the need for classical post-processing methods that search for feasible solutions in the vicinity of infeasible solutions generated by the PQC (cf.~\ref{subsec:outlook}): None of the bit-vectors generated by both simulation and real hardware fulfil all the constraints on the security-account balances (cf. equation~\ref{eq:constraints}). Alternatively, the cost penalty $\lambda$ could be increased to put even higher priority on the balance constraints relative to the maximization of the number of transactions.

\begin{figure} 
    \begin{subfigure}[t]{0.47\textwidth}
        \includegraphics[width =1.0\linewidth]{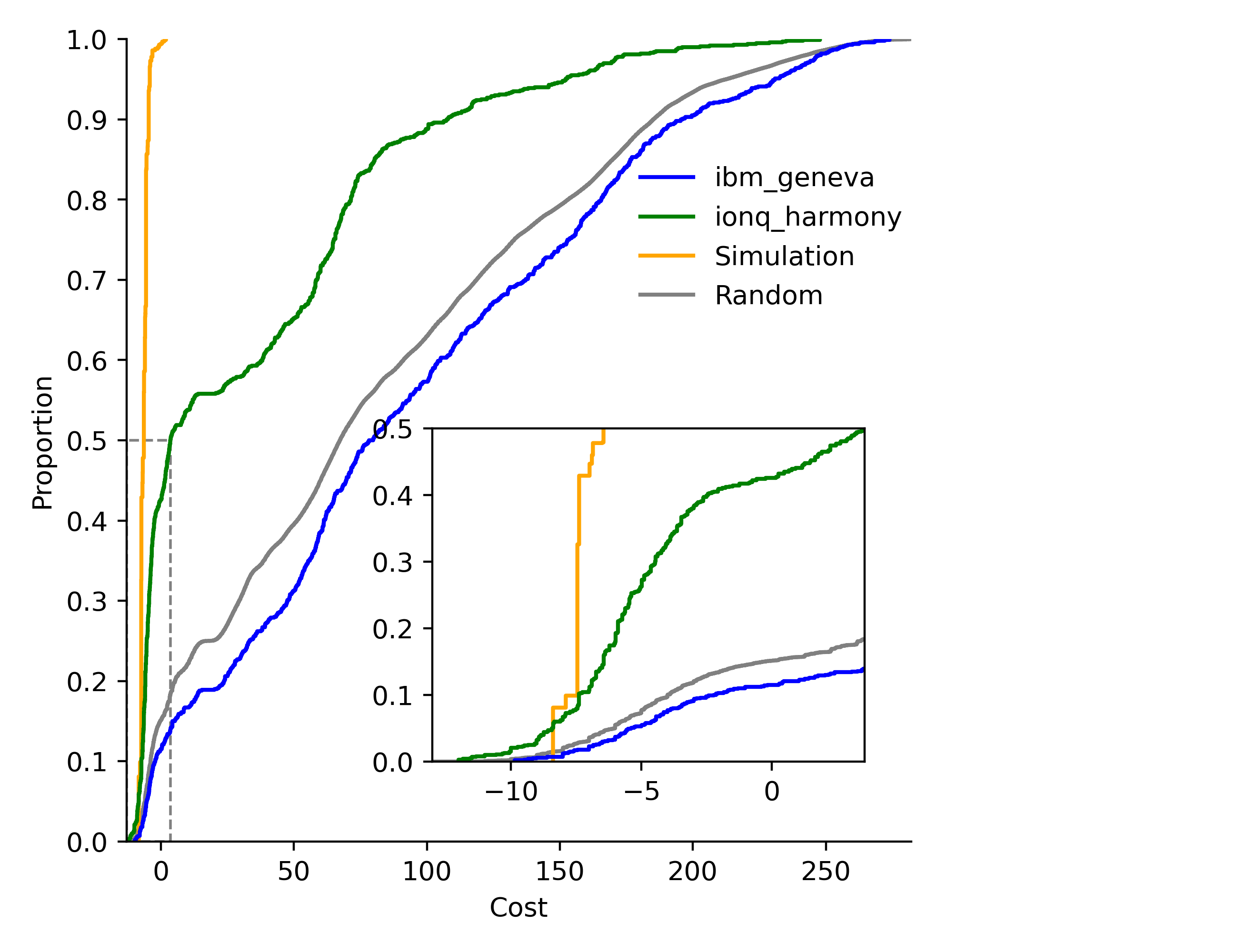}
        \caption{Register-preserving ansatz}
        \label{fig:cdf_reg-pres-HW}
    \end{subfigure}
    \begin{subfigure}[t]{0.47\textwidth}
        \includegraphics[width =1.0\linewidth]{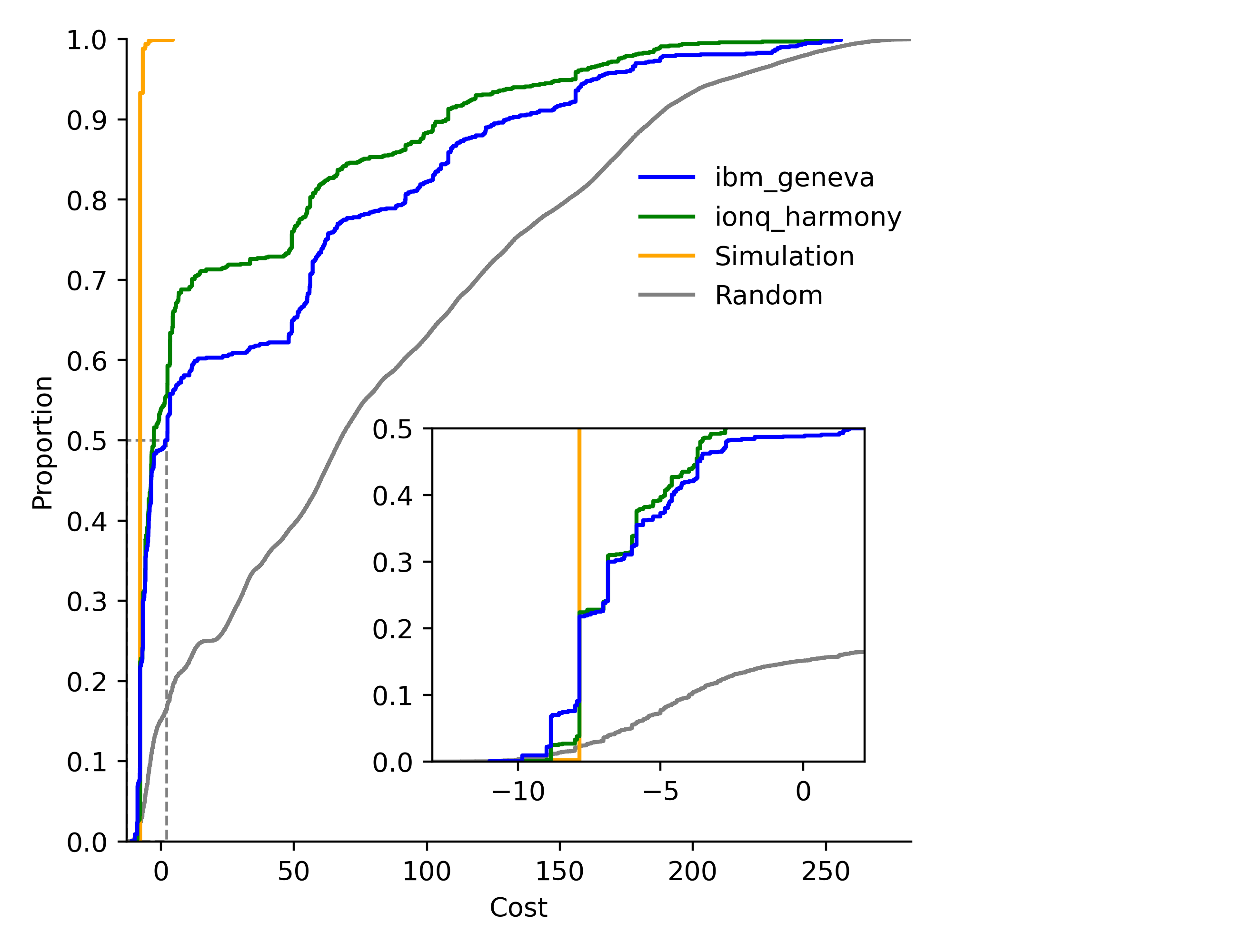}
        \caption{Hardware-efficient ansatz}
        \label{fig:cdf_hwe-HW}
    \end{subfigure}
    \caption[Hardware results 16 Transactions]{\colorbox{lightgray!30!white}{16 Transactions, 10 parties, 6 qubits, $d=4$, DESC, $n_a = 4$, disjoint covering} | Empirical cumulative distribution function (cf.~figure~\ref{fig:cdf}
    ) of cost (not normalized) over bit-vectors generated by a trained PQC. Executed on a simulator (Simulation), IBMQ quantum computer (\texttt{ibm\_geneva}) and IonQ quantum computer (\texttt{ionq\_harmony}). Uniformly random bit-vectors are shown by a grey line (Random). For each generator, $1000$ bit-vectors were drawn from $2.4 \times 10^4$ shots. Each subfigure corresponds to a single optimization run in which the PQC was optimized using gradient descent. Both runs produce narrow distributions around a fixed bit-vector, resulting in a steep curve (especially in the noise-free simulation). \\
    }\label{fig:cdf_HW}    
    
\end{figure}

\begin{figure}
    \centering
    \begin{subfigure}[t]{0.52\textwidth}
        \includegraphics[width =1.0\linewidth]{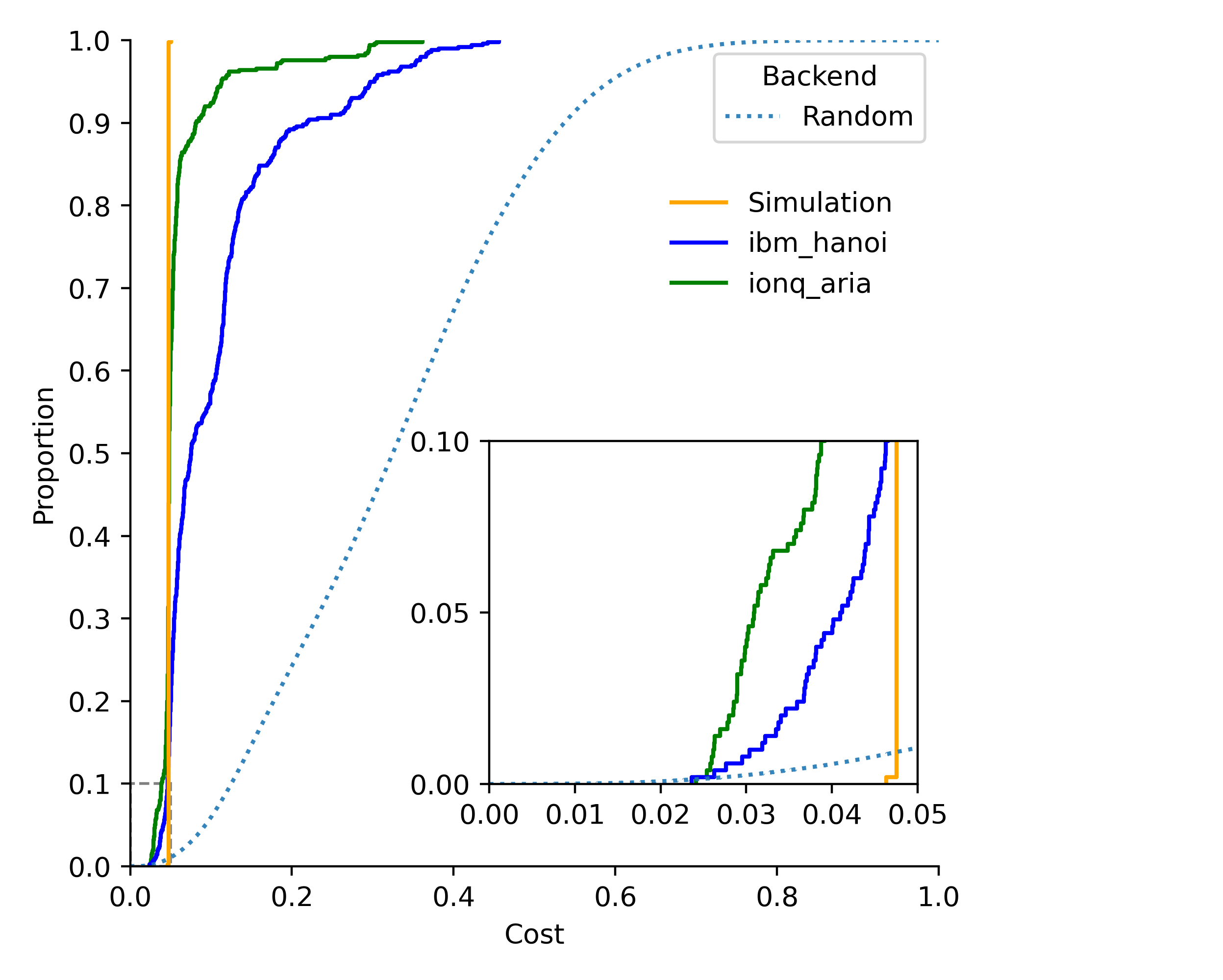}
        \caption{Empirical cumulative distribution function (cf. figure~\ref{fig:cdf}). To generate the blue dotted line, showing the distribution of uniformly random bit-vectors, we generated $10^6$ such vectors, the cost of which was used to normalize the x-axis.}
        \label{fig:cdf_HW_128TRS}
    \end{subfigure}
        \begin{subfigure}[t]{0.45\textwidth}
        \includegraphics[width =1.0\linewidth]{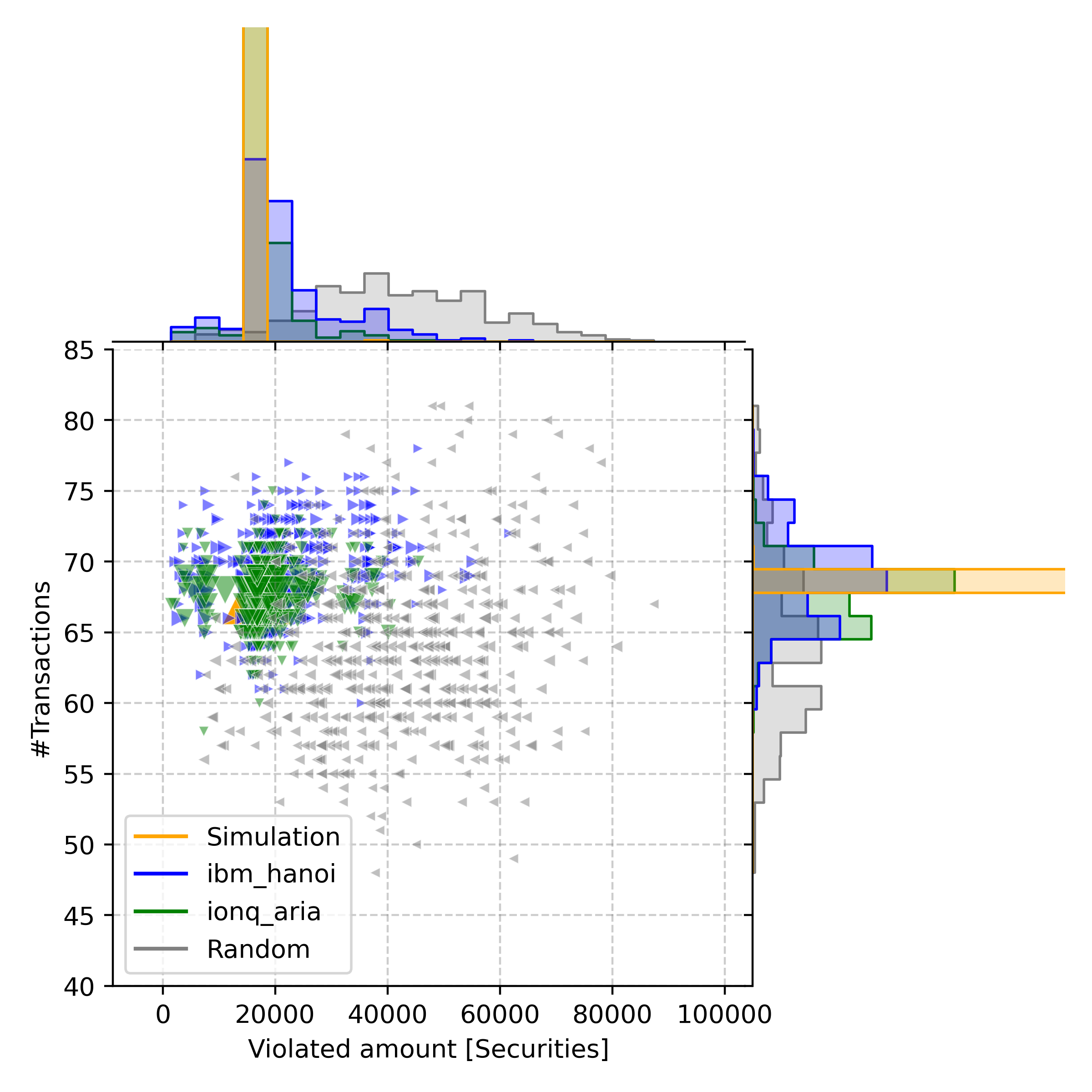} 
        \caption{Scatter-plot of sampled bit-vectors. Scatter-size increases with degeneracy. The x-axis shows the negative sum of final security-account balances over all parties with negative balances upon settlement according to the corresponding bit-vector. The number of transactions is the $L_1$ norm of the corresponding bit-vector. 
        }
        \label{fig:ViolxTRS_HW_128TRS}
    \end{subfigure}
    \caption[Hardware results 128 Transactions]{\colorbox{lightgray!30!white}{128 transactions, 41 parties, 19 qubits, hardware-efficient ansatz, $d=1$, DESC, $n_a = 16$, disjoint covering} | Visualization of results for 128 transactions on the \texttt{ibm\_hanoi} and \texttt{ionq\_aria} backends ($2.4 \times 10^4$ shots, 500 bit-vectors).}
    \label{fig:ionq_ibm_128trs_HW}
\end{figure}


\newpage 

\section{Conclusion and Outlook}
\label{04_Conclusion}

Increasing the scope of tractable problems and benchmarking with industry data is important to gauge the applicability of heuristics-reliant variational quantum algorithms to optimization and to find promising applications. In this work, we extended the qubit-efficient encoding in~\cite{tan_qubit-efficient_2021} by providing explicit formulas of the cost objective and its gradient for arbitrary number of ancilla qubits. We introduced a new ansatz for uniform register sampling. We argue that register-preserving ansätze have the additional benefits of numerical stability, shot-reductions and selective sampling of individual registers, and expressing the cost estimator as a Hermitian observable.

We demonstrate our methods on mixed binary optimization problems arising from financial transaction settlement ~\cite{braine_quantum_2019}, benchmarking problems of up to 128 transactions and 41 parties constructed from transaction data provided by a regulated financial exchange.
We also showed how the optimal slack variables can be obtained without the need for an outer loop optimization. We observed that our qubit-efficient methods outperformed standard QAOA, even when executed on quantum backends. Our proposed register-preserving ansatz stood out as best in many of the instances considered.

\label{subsec:outlook}


\paragraph{Post-processing} 
While not explored in this work, post-processing by projecting sampled solutions to valid bit-vectors fulfilling all problem constraints
may be a necessity. One possible method to do so includes projecting a solution bit-vector generated by the QS to the best bit-vector \emph{in the vicinity} that fulfils all constraints. Restricting the search to a ball of Hamming distance $\leq k$, solutions can be sampled at the asymptotic runtime of $O(I^k)$ (as opposed to $2^I$ for a full search). While without guarantees for the optimality or even existence of a close valid solution, one may hope that if the QS provides solutions of high quality, only small adjustments are needed to obtain a good solution which adheres to all constraints. This search can be refined with heuristics, e.g.~by only adjusting transactions involving parties (and potentially their k-nearest-neighbours on the graph) whose balance constraints are violated.




\paragraph{Method exploration} 
Overall, further exploration of different ancilla-register-mappings, variational (register-preserving) ansätze, optimization algorithms and (scaling of) meta-parameters such as circuit depth, penalty terms and step-size is warranted to validate and refine the qubit-compression approach presented in this work. 
For example, how the restrictions given by theorem~\ref{theorem:reg_preserving} regarding the register-preserving ansatz can best be extended in practice, e.g. by changing the computational basis and keeping track of phases on the ancilla qubits (cf.~\ref{lemma:reg_pres_sometimes}), is still an open question. 
Another important consideration is the lack of correlation between the individual measurements in the sampling algorithm used in this work. Exploring sample rejection or the addition of hidden layers to the ansatz provides one avenue to extend the probability distributions of $\myvec x$ which our PQC parameterizes.

\medskip

Our methods, despite being tested on synthetic problems, demonstrated the broad applicability of quantum algorithms beyond small toy examples. Witnessing advantages of our methods over classical solvers would require a comparison to state-of-the-art classical solvers on problem instances faced in real scenarios. 
Most of our methods are directly applicable beyond the transaction settlement problem to any QUBO problem with linear inequality constraints, setting them apart from other qubit-efficient methods to the best of our knowledge and making them suitable to tasks beyond settling financial transactions.\footnote{For a list of examples see~\cite{punnen_quadratic_2022} ch.~2.2(.1).}

\newpage

\section*{Declarations}
\label{sec:declarations}
\addcontentsline{toc}{section}{\nameref{sec:declarations}}

\subsection*{Ethical Approval and Consent to participate}
Not applicable.

\subsection*{Consent for publication}
The authors consent to publication by EPJ Quantum Technology.

\subsection*{Availability of supporting data}
Given permission by the regulated stock exchange, anonymized transaction data samples used and/or analysed during the current study are available from the corresponding author on reasonable request.

\subsection*{Competing interests and Authors' contributions}
E.H. wrote the main manuscript and conducted the underlying analysis, excluding simulations of QAOA which were conducted by B.T., who also contributed to the research design and revised the manuscript. D.A. and P.G. coordinated the project and facilitated contact with the regulated financial exchange. All authors reviewed the manuscript and have no competing interests to declare.

\subsection*{Acknowledgements and Funding}

We thank Daniel Leykam for his valuable comments and suggestions. This research is supported by the National Research Foundation, Singapore and A*STAR (\#21709) under its CQT Bridging Grant and Quantum Engineering Programme (NRF2021-QEP2-02-P02) and by EU HORIZON-Project101080085—QCFD. We acknowledge IBM Quantum, IonQ and Amazon Web Services.




\newpage

\FloatBarrier

\printbibliography

\appendix
\section{Appendix: Additional Figures}\label{app:figs}

\begin{figure}[h!]
    \centering
    \includegraphics[width =0.5\linewidth]{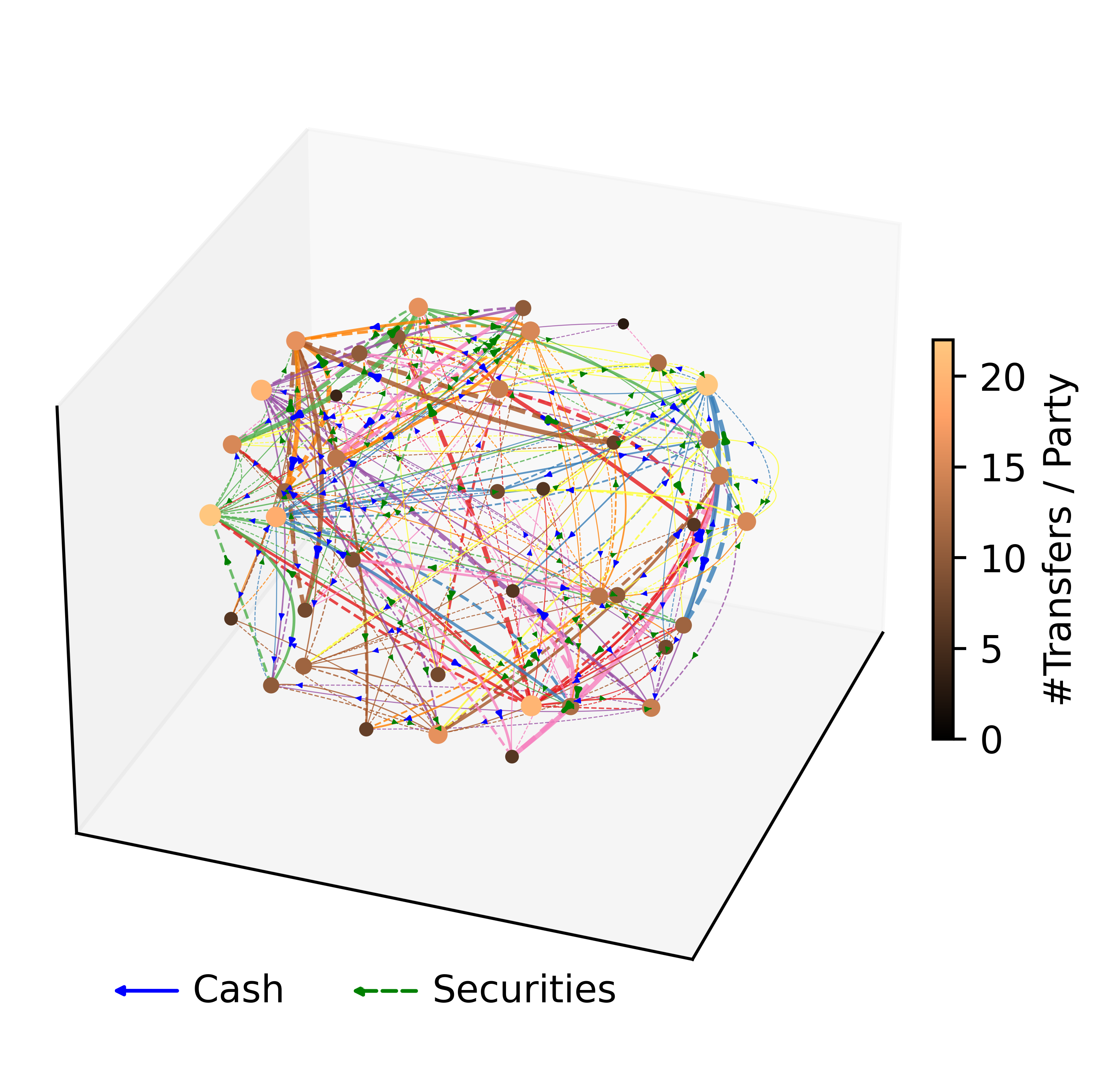}
    \caption[Visualization 128 transactions]{\colorbox{lightgray!30!white}{128 Transactions (each with security and optional cash payment), 41 parties} | Parties and transactions are shown similar to figure~\ref{fig:TRS_settlement_example} for a settlement problem generated from data of a regulated financial exchange (same settlement problem as figure~\ref{fig:ionq_ibm_128trs_HW}). Transfers of securities (cash) are visualized by dotted (solid) edges and green (blue) arrow-heads. The thickness of the edges scales proportional to the transaction volume. The nodes represent parties and are colored and scaled according to the number of in- and outgoing transfers. The edge-coloring corresponds to the register-mapping, i.e.~the sets $A_i \in \mathcal{A}$. For $n_a = 16$ ancilla qubits and $n_r = 3$ register qubits used to generate the register-mapping, we have $2^3=8=|\mathcal{A}|$ different colors.
    }
    \label{fig:graph_settlement_instructions_128}
\end{figure}


\FloatBarrier

\section{Appendix: Proofs}\label{app:proofs}

In this section, we provide formal proofs of claims made in the main text:
\smallskip

\paragraph{Register-preserving ansatz}

\lemmareguniform*

\begin{proof} ~\\

\enquote{$\implies$}: Given a register-uniform state $\ket{\psi} = \frac{1}{\sqrt{N_r}}\sum_{r = 1}^{N_r}\ket{\phi_r}_{\text{anc}}\otimes\ket{r}_{\text{reg}}$, direct calculation shows 
\begin{equation}
    \bra{r}\text{tr}_{\mathcal{H}_{\text{anc}}}(\ket{\psi}\bra{\psi})\ket{r}_{\text{reg}} = \frac{1}{N_r}\text{tr}(\ket{\phi_r}\bra{\phi_r}_{\text{anc}}) = \frac{1}{N_r}.
\end{equation}

\enquote{$\impliedby$}: Given $\ket{\psi}$ with $\bra{r}\text{tr}_{\mathcal{H}_{\text{anc}}}(\ket{\psi}\bra{\psi})\ket{r}  = \frac{1}{N_r}$, write $\ket{\psi}$ in the computational basis as 
\begin{equation}
    \ket{\psi} = \sum_{i = 1}^{N_a}\sum_{r=1}^{N_r}\lambda_{ir} \ket{i}_{\text{anc}}\otimes\ket{r}_{\text{reg}} = \frac{1}{\sqrt{N_r}}\sum_{r=1}^{N_r}\underbrace{\left[\sqrt{N_r}\sum_{i = 1}^{N_a}\lambda_{ir} \ket{i}_{\text{anc}}\right]}_{=:\ket{\phi_r}_{\text{anc}}}\otimes\ket{r}_{\text{reg}}
\end{equation}
Then
\begin{equation}
    \braket{\phi_r | \phi_r}_{\text{anc}} = \text{tr}(\ket{\phi_r}\bra{\phi_r}_{\text{anc}}) = N_r\text{tr}(\mathds{1}\otimes\bra{r}_{\text{reg}}\ket{\psi}\bra{\psi}\mathds{1}\otimes\ket{r}_{\text{reg}}) = N_r\bra{r}\text{tr}_{\mathcal{H}_{\text{anc}}}(\ket{\psi}\bra{\psi})\ket{r}  = 1
\end{equation}
showing that $\ket{\psi}$ is register-uniform.

\end{proof}

\theoremregpreserving*

\begin{proof}
We will proof in order $(i)\implies (ii)\implies (iii)\implies (i)$.

\smallskip

$(i) \implies (ii)$:
The idea in this part is the following: Since the ancilla states for each register are arbitrary, no \enquote{mixing} between different registers is allowed as this would result in uncontrollable superpositions on the ancilla subsystems causing a loss of \enquote{normalization}.
We will first show the following weaker claim:

\textbf{Claim 1:} $\forall\ket{\phi_r}_{\text{anc}} \in \mathcal{H}_{\text{anc}}$, $1\leq r\leq N_r$, $\exists\ket{\Tilde{\nu}_r}_{\text{anc}} \in \mathcal{H}_{\text{anc}}$ and $1\leq \Tilde{r}\leq N_r$, s.t. $U \ket{\phi_r}_{\text{anc}}\ket{r}_{\text{reg}} = \ket{\Tilde{\nu}_r}_{\text{anc}}\ket{\Tilde{r}}_{\text{reg}}$ 

\emph{Proof of Claim 1:} Assume this was not the case, i.e. $\exists\ket{\phi_s}_{\text{anc}} \in \mathcal{H}_{\text{anc}}, 1\leq s\leq N_r$ s.t.
\begin{equation}
    U\ket{\phi_s}_{\text{anc}}\ket{s}_{\text{reg}} = \lambda \ket{\Tilde{\nu}_s}_{\text{anc}}\ket{\Tilde{s}}_{\text{reg}} + \Bar{\lambda}\ket{\Tilde{\psi}}
    \label{eq:claim1_not}
\end{equation}
for some $\lambda, \ket{\Tilde{\nu}_s}_{\text{anc}}, \Tilde{s}, \ket{\Tilde{\psi}}$ with $0<|\lambda| < 1$ and $\mathds{1}_{\text{anc}}\otimes\bra{\Tilde{s}}_{\text{reg}}\ket{\Tilde{\psi}} = 0$. We denoted $\Bar{\lambda}:= \sqrt{1-|\lambda|^2}$. We note at this point, that adding a complex phase to $\ket{\phi_s}_{\text{anc}}$ merely adds the same phase to $\lambda$ and $\ket{\Tilde{\psi}}$ but otherwise does not change equation~\ref{eq:claim1_not}.

We now apply $U$ to a register-uniform state containing $\ket{\phi_s}_{\text{anc}}\ket{s}_{\text{reg}}$, resulting in
\begin{align}
    \ket{\Tilde{\phi}_{\Tilde{s}}}_{\text{anc}}&:= \mathds{1}_{\text{anc}}\otimes\bra{\Tilde{s}}_{\text{reg}} U \sum_{r = 1}^{N_r}\ket{\phi_r}_{\text{anc}}\otimes\ket{r}_{\text{reg}} \\
    &= \lambda \ket{\Tilde{\nu}_s}_{\text{anc}}+\mu\ket{\Tilde{\eta}}_{\text{anc}}
\end{align}
where $\mu\ket{\Tilde{\eta}}_{\text{anc}}=\mathds{1}_{\text{anc}}\otimes\bra{\Tilde{s}}_{\text{reg}}U \sum_{\substack{r = 1\\r\neq s}}^{N_r}\ket{\phi_r}_{\text{anc}}\otimes\ket{r}_{\text{reg}}$.
We will fix $\ket{\phi_r}_{\text{anc}}$ for $ r\neq s$ later. Using that $U$ is register-preserving, we have
\begin{equation}
    1 = \braket{\Tilde{\phi}_{\Tilde{s}}|\Tilde{\phi}_{\Tilde{s}}}_{\text{anc}} = |\lambda|^2+|\mu|^2+\lambda^*\mu \braket{\Tilde{\nu}_s|\Tilde{\eta}}_{\text{anc}}+\lambda\mu^* \braket{\Tilde{\eta}|\Tilde{\nu}_s}_{\text{anc}}.
\end{equation}
If we can show that $\mu\braket{\Tilde{\nu}_s|\Tilde{\eta}}_{\text{anc}} \in \mathbb{R}_{>0}$, then we can choose the phase of $\ket{\phi_s}_{\text{anc}}$ such that
\begin{itemize}
    \item Case 1: $\lambda = |\lambda|$
    \item Case 2: $\lambda = i|\lambda|$
\end{itemize}
which results in the contradiction
\begin{equation}
    |\lambda|^2+|\mu|^2+2|\lambda|\mu\braket{\Tilde{\nu}_s|\Tilde{\eta}}_{\text{anc}} = 1 = |\lambda|^2+|\mu|^2.
\end{equation}
It remains to show, that we can choose $\ket{\phi_r}_{\text{anc}}, r\neq s$ such that $\mu\braket{\Tilde{\nu}_s|\Tilde{\eta}}_{\text{anc}} \neq 0$ and real.

For this consider 
\begin{equation}
    U^{-1} \ket{\Tilde{\nu}_s}_{\text{anc}}\ket{\Tilde{s}}_{\text{reg}} = \frac{1}{\lambda}\ket{\phi_s}_{\text{anc}}\ket{s}_{\text{reg}}-\frac{\Bar{\lambda}}{\lambda} U^{-1} \ket{\Tilde{\psi}}
\end{equation}

\textbf{Claim 2:} $\exists w\neq s \in \{1,\ldots, N_r\}: \mathds{1}_{\text{anc}}\otimes\bra{w}_{\text{reg}} U^{-1} \ket{\Tilde{\nu}_s}_{\text{anc}}\ket{\Tilde{s}}_{\text{reg}} \neq 0$ 

\emph{Proof of Claim 2:} Assume this was not the case. Then
\begin{equation}
    U^{-1} \ket{\Tilde{\nu}_s}_{\text{anc}}\ket{\Tilde{s}}_{\text{reg}} = \frac{1}{\lambda}\ket{\phi_s}_{\text{anc}}\ket{s}_{\text{reg}}+\sigma \ket{\theta_s}_{\text{anc}}\ket{s}_{\text{reg}}
\end{equation}
Hence span$\{\ket{\Tilde{\nu}_s}_{\text{anc}}\ket{\Tilde{s}}_{\text{reg}}\} \cap \text{Im}[U_{|\text{span}\{\mathcal{H}_{\text{anc}}\otimes\ket{r}_{\text{reg}}\}}]=\emptyset~\forall r\neq s$.
In particular for $r\neq s$
\begin{equation}
    \text{dim}\left(\text{Im}[U_{|\text{span}\{\mathcal{H}_{\text{anc}}\otimes\ket{r}_{\text{reg}}\}}]\cap \text{span}\{\mathcal{H}_{\text{anc}}\otimes\ket{\Tilde{s}}_{\text{reg}}\}\right)\leq \text{dim}(\mathcal{H}_{\text{anc}})-1 = \text{dim}(\text{Im}[U_{|\text{span}\{\mathcal{H}_{\text{anc}}\otimes\ket{r}_{\text{reg}}\}}])-1.
\end{equation}
Therefore
\begin{equation}
    \forall r\neq s~ \exists \ket{\phi_r}_{\text{anc}} \text{ s.t } \mathds{1}_{\text{anc}}\otimes\bra{\Tilde{s}}_{\text{reg}} U \ket{\phi_r}_{\text{anc}}\otimes\ket{r}_{\text{reg}} = 0
\end{equation}
which implies
\begin{equation}
    \mathds{1}_{\text{anc}}\otimes\bra{\Tilde{s}}_{\text{reg}} U \sum_{r = 1}^{N_r}\ket{\phi_r}_{\text{anc}}\otimes\ket{r}_{\text{reg}}  = \lambda \ket{\Tilde{\nu}_s}_{\text{anc}}\ket{\Tilde{s}}_{\text{reg}}
\end{equation}
which is not normalized ($|\lambda|<1$), in contradiction with $U$ being register-preserving.

Therefore claim 2 holds and it exists $\ket{\phi_w}_{\text{anc}}\otimes\ket{w}_{\text{reg}}$ with $w\neq s$ such that
\begin{equation}
    \bra{\Tilde{\nu}_s}_{\text{anc}}\otimes\bra{\Tilde{s}}_{\text{reg}} U \ket{\phi_w}_{\text{anc}}\otimes\ket{w}_{\text{reg}} \neq 0.
    \label{eq:overlap_nons}
\end{equation}
We can furthermore fix the phase of $\ket{\phi_w}_{\text{anc}}$ such that the overlap in~\ref{eq:overlap_nons} is real and positive.
We can choose all other $\ket{\phi_r}_{\text{anc}}$ for $r \neq w, s$ arbitrary with the condition that if the overlap $\bra{\Tilde{\nu}_s}_{\text{anc}}\otimes\bra{\Tilde{s}}_{\text{reg}} U \ket{\phi_r}_{\text{anc}}\otimes\ket{r}_{\text{reg}}$ is non-zero, we adjust the phase so that the overlap is real and positive.
From this it follows
\begin{equation}
    \mu\braket{\Tilde{\nu}_s|\Tilde{\eta}}_{\text{anc}} = \bra{\Tilde{\nu}_s}_{\text{anc}}\otimes\bra{\Tilde{s}}_{\text{reg}}U \sum_{\substack{r = 1\\r\neq s}}^{N_r}\ket{\phi_r}_{\text{anc}}\otimes\ket{r}_{\text{reg}} \in \mathbb{R}_{>0}
\end{equation}

which concludes claim 1.

To prove the more restrictive statement $(ii)$, it remains to show:
\begin{enumerate}[label=\alph*)]
    \item $\ket{r}_{\text{reg}}\mapsto \ket{\Tilde{r}}_{\text{reg}}$ is a well-defined bijective function, $\Tilde{r} = f(r)$, $f: \{1,\ldots,N_r\}\to\{1,\ldots,N_r\}$ 
    \item $\ket{\phi_r}_{\text{anc}}\mapsto \ket{\Tilde{\nu}_r}_{\text{anc}} = \mathds{1}\otimes \bra{f(r)}U \ket{\phi_r}_{\text{anc}}\ket{r}_{\text{reg}} $ is unitary $\forall \ket{\phi_r}_{\text{anc}}\in \mathcal{H}_{\text{anc}}, ~ 1\leq r\leq N_r$
\end{enumerate}
To show that f is a well-defined function, we need to show in addition to claim 1 that the registers are mapped independently of the state of the ancilla. Assume this was not the case,~i.e. $\exists r, \ket{\phi_r^1}_{\text{anc}}, \ket{\phi_r^2}_{\text{anc}}$ s.t.
\begin{equation}
    U \ket{\phi_r^i}_{\text{anc}}\otimes\ket{r}_{\text{reg}} = \ket{\Tilde{\nu}_r^i}_{\text{reg}}\otimes \ket{\Tilde{r}^i}_{\text{reg}}
\end{equation}
with $\Tilde{r}^1 \neq \Tilde{r}^2$. Consider two register-uniform states which are identical besides one containing $\ket{\phi_r^1}_{\text{anc}}\otimes\ket{r}_{\text{reg}}$ and the other $\ket{\phi_r^2}_{\text{anc}}\otimes\ket{r}_{\text{reg}}$. Then due to linearity, one of the two states is not mapped to a register-uniform state under $U$. This is a contradiction and hence f is a well-defined function. Because U is register-preserving, f must be surjective and is hence bijective.

To show b) we note that for any $r$ the restriction $U: \text{span}\{\mathcal{H}_{\text{anc}}\otimes \ket{r}\} \to \text{span}\{\mathcal{H}_{\text{anc}}\otimes \ket{f(r)}\}$ is unitary.

$(ii) \implies (iii)$: Given $U_r$ and $f(\circ)$ as in $(ii)$, we can directly define a sequence of conditioned unitaries and a permutation on the registers. 

For any $r$ we can condition the unitary $U_r$ acting on $\mathcal{H}_{\text{anc}}$ on the bitvector-representation of $r$. Let us call the resulting unitary $CU_r$. Define the permutation $P: \mathcal{H}_{\text{reg}}\ni \ket{r} \mapsto \ket{f(r)} \in \mathcal{H}_{\text{reg}}$. Then
\begin{equation}
    U = \mathds{1}_{\text{anc}}\otimes P \prod_{r = 1}^{N_r}CU_r
\end{equation}

$(iii) \implies (i)$: As the set of register-preserving unitaries is closed under composition, it suffices to show that basis-permutations on the registers and unitaries on the ancillas conditioned on register-states are both register-preserving. Both can be verified explicitly by tracking the action on a state of the form $\frac{1}{\sqrt{N_r}}\sum_{r = 1}^{N_r}\ket{\phi_r}_{\text{anc}}\otimes\ket{r}_{\text{reg}}$ and observing that the normalization of the individual terms in the sum is not compromised. Permutations on the register merely reassign the ancilla states to different addresses while conditional unitaries on the ancillas apply unitary transformations on the ancilla-states for a subset of registers specified through the conditioning.

\end{proof}

\lemmaregpreservingsometimes*

\begin{proof}
    We will show the claim by demanding $U\ket{\psi}$ adhere to definition~\ref{def:reg_uniform}. 
    
    Projecting $U\ket{\psi}$ on a register $s$, we obtain
    \begin{align}
        \mathds{1}_{\text{anc}}\otimes \bra{s}_{\text{reg}} \mathds{1}\otimes U_{\text{reg}}\frac{1}{\sqrt{N_r}}\sum_{r = 1}^{N_r}\ket{\phi_r}_{\text{anc}}\otimes\ket{r}_{\text{reg}} = \frac{1}{\sqrt{N_r}}\ket{\Tilde{\phi}_s}_{\text{anc}}
    \end{align}
    and therefore 
    \begin{align}
        \braket{\Tilde{\phi}_s|\Tilde{\phi}_s}_{\text{anc}} = 1 \iff
        \sum_{r_1, r_2 = 1}^{N_r}\braket{\phi_{r_1}|\phi_{r_2}}_{\text{anc}} \bra{r_1}U_{\text{reg}}^\dagger\ket{s}_{\text{reg}}\bra{s}U_{\text{reg}}\ket{r_2}_{\text{reg}} = 1
    \end{align}
    We can split the sum to obtain (using the assumption of real-valued ancilla overlap)
    \begin{align}
        1 = \underbrace{\sum_{r = 1}^{N_r}\bra{r}U_{\text{reg}}^\dagger\ket{s}_{\text{reg}}\bra{s}U_{\text{reg}}\ket{r}_{\text{reg}}}_{= 1}
        +\sum_{1 \leq r_1<r_2 \leq N_r}\braket{\phi_{r_1}|\phi_{r_2}}_{\text{anc}} \underbrace{\left[\bra{r_1}U_{\text{reg}}^\dagger\ket{s}_{\text{reg}}\bra{s}U_{\text{reg}}\ket{r_2}_{\text{reg}}+\bra{r_2}U_{\text{reg}}^\dagger\ket{s}_{\text{reg}}\bra{s}U_{\text{reg}}\ket{r_1}_{\text{reg}}\right]}_{ = 2~ \text{Re}\left[(U^\dagger)_{r_1s}U_{sr_2}\right]}
        \label{eq:norm_condition_ref}
    \end{align}    
    where the first term equals 1 due to $U_{\text{reg}}\ket{r}_{\text{reg}}$ forming an orthonormal basis of $\mathcal{H}_{\text{reg}}$. As this equations has to hold for all sets $\{\ket{\phi_r}_{\text{anc}}\}_{r=1}^{N_r}$ which have real-valued overlap (hence we can engineer them such that exactly one chosen term $\braket{\phi_{r_1}|\phi_{r_2}}_{\text{anc}}$ in the sum is non-zero) and for all $s$, equation~\ref{eq:norm_condition_ref} is equivalent to $\text{Re}\left[(U^\dagger)_{r_1s}U_{sr_2}\right]=0~\forall r_1,r_2,s\in \{1,\ldots, N_r\} \text{ with } r_1<r_2$.
\end{proof}

The resulting state is register-uniform due to cancellation of phases. While this only holds for real-valued inner product between ancilla-states, it raises the question whether similar results are possible more generally if a record of the phases of the ancilla states is kept. Furthermore, we did not explore the possibility of allowing different register bases for the input and output state in the definition of a register-preserving unitary.

\paragraph{Transaction settlement problem}

For some combinatorial graph problems such as MaxCut, the corresponding QUBO matrix $Q$ is in simple correspondence with the graph and its adjacency matrix. This is not the case here, where deciding which transactions to settle corresponds to choosing a subset of \emph{edges} (not nodes, cf. figure~\ref{fig:TRS_settlement_example}). Quadratic terms beyond the graph connectivity are common for QUBO problems which incorporate constraints as quadratic penalties. This increase the number of non-zero off-diagonal elements of $Q$ and hence limits the applicability of many NISQ-QS as stressed in section~\ref{subsec:quantum_optimization_intro}. For transaction settlements, we can relate the number of non-zero elements per row of $Q$ as follows:
\begin{restatable}[]{lm}{lemmaconnectivity}
\label{lemma:connectivity}
    Given a transaction settlement with $A, ~\myvec b(\myvec s)$ as in equations~\ref{eq:QUBO_trs_A} and~\ref{eq:QUBO_trs_b}, represented by a graph $G$ with parties as $K$ nodes $\mathcal{K}$ and transactions as $I$ edges $\mathcal{I}$ connecting the transacting parties, then: The average number of non-zero entries per row of the matrix $Q = A + \textrm{Diag}(\myvec b(\myvec s))$ is bounded by
    \begin{equation}
        \mathbb{E}_{\mathcal{I}}\left [\sum_{j = 1}^I\delta_{\left\{Q_{ij}\neq 0 \right\}}\right ] \leq 4\frac{I}{K}+\frac{K}{I}\mathbb{V}_{\mathcal{K}}\left[N_k\right]-1
    \end{equation}
    where $N_k$ is the number of edges connected to node $k\in \mathcal{K}$.
\end{restatable}
For \emph{d-regular graphs} $\mathbb{V}_{\mathcal{K}}\left[N_k\right]=0$.

\begin{proof}
    We are looking for an upper-bound of the average number of non-zero elements in the rows of $Q$. For this assume the diagonal elements are all non-zero. For the off-diagonal elements we only need to consider the contributions of $A = -\lambda V V^T$. As $V_{il} := (\myvec v_{ik(l)})_{j(l)}$ is only non-zero if transaction $i$ changes balance $j(l)$ of party $k(l)$, $A_{ij}$ may only be non-zero if edges $i$ and $j$ share a node.
    Therefore
    \begin{equation}\label{eq:conn_rewr}
        \mathbb{E}_{\mathcal{I}}\left [\sum_{j = 1}^I\delta_{\left\{Q_{ij}\neq 0 \right\}}\right ] = \frac{1}{I}\sum_{i=1}^I\left [\sum_{j = 1}^I\delta_{\left\{Q_{ij}\neq 0 \right\}}\right ] \leq  \frac{1}{I}\sum_{i=1}^I\left [\underbrace{\sum_{j \in \mathcal{I}}\delta_{\left\{\text{Transaction } i \text{\&} j \text{ share a node} \right\}}}_{=:a_i}\right ].
    \end{equation}
    Let $k_1(i), k_2(i) \in \mathcal{K}$ such that $i = (k_1(i), k_2(i))$\footnote{In the case of an undirected graph we can fix any ordering}. Then \[a_i \leq N_{k_1(i)}+N_{k_2(i)}-1\] where $-1$ comes from double-counting the edge $i$ connecting $k_1(i)$ with $k_2(i)$.
    Plugging this into equ.~\ref{eq:conn_rewr} we get
    \begin{align}
        \mathbb{E}_{\mathcal{I}}\left [\sum_{j = 1}^I\delta_{\left\{Q_{ij}\neq 0 \right\}}\right ] &\leq -1 + \frac{1}{I} \sum_{k=1}^K N_k\left[\sum_{i \in \mathcal{I}}\delta_{\left\{i_0 = k\right\}}+\delta_{\left\{i_1 = k\right\}}\right] \\
        &= -1 + \frac{K}{I}\left[\frac{1}{K}\sum_{k=1}^K N_k^2\right] = -1+ \frac{K}{I}\left[\mathbb{E}_{\mathcal{K}}[N_k]^2+\mathbb{E}_{\mathcal{K}}[N_k^2]-\mathbb{E}_{\mathcal{K}}[N_k]^2\right]\\
        &=-1+\frac{K}{I}\left[\left(2\frac{I}{K}\right)^2+\mathbb{V}_{\mathcal{K}}[N_k]\right]\\
        &=4\frac{I}{K}+\frac{K}{I}\mathbb{V}_{\mathcal{K}}[N_k]-1
    \end{align}
\end{proof}

\section[Appendix: Cost estimator and marginal probability]{Appendix: Cost estimator and marginal probability distribution obtained from sampling algorithm}\label{app:marg_prob_from_sampling}

In this section, we give explicit formulas for the estimators in section~\ref{subsec:cost} and their gradients. For the special case of a disjoint covering and greedy sampling algorithm, 
we rigorously prove the heuristic estimators used in the cost objective.

\subsection{Cost estimator}\label{subapp:cost_estimator_formula}

To express $\hat{p}_i$ and $\hat{p}_{ij}$ through a set of measurements $\mathcal{M}=\left\{\myvec m_1, \ldots \myvec m_{n_{\text{shots}}}\right\}$ in the computational basis, let $\Tilde{\myvec x}(\myvec m)$ refer to the \enquote{bit}-vector $\in \{-1,0,1\}^I$ with the subset $A_{r(\myvec m)}$ of entries fixed to $b(\myvec m)_1 \ldots b(\myvec m)_{n_a}$ as described in section~\ref{subsec:qubit_compression} 
and all other bits set to $-1$. Define $l_r(\circ): A_r \to \{1,\ldots,n_a\}$ through $A_r[l_r(i)]=i~\forall i \in A_r$ (\enquote{bit $i$ is mapped to $l_r(i)^{\text{th}}$ ancilla bit of register $r$}).
The estimators are given as follows: 
\begin{align}
    \hat{p}_i(\myvec \theta) 
    &:= \frac{\sum_{\myvec m \in \mathcal{M}}\delta_{\Tilde{x}_i(\myvec m),1}}{\sum_{\myvec m \in \mathcal{M}}(1-\delta_{\Tilde{x}_i(\myvec m),-1})} 
    \label{eq:p_hat_pauli}\\
    \hat{p}_{ij}(\myvec \theta) 
    &:= (1-\hat{\mu}_{ij})\hat{q}_{ij}+\hat{\mu}_{ij}\hat{p}_i\hat{p}_j \label{eq:p_ij_hat_pauli}\\
    \text{where:  } \hat{q}_{ij}(\myvec \theta)
    &:= \frac{\sum_{\myvec m \in \mathcal{M}}\delta_{\Tilde{x}_i(\myvec m),1}\delta_{\Tilde{x}_j(\myvec m),1}}{\sum_{\myvec m \in \mathcal{M}}(1-\delta_{\Tilde{x}_i(\myvec m),-1})(1-\delta_{\Tilde{x}_j(\myvec m),-1})} \\
    \label{eq:q_hat_pauli}\\
    \hat{\mu}_{ij}(\myvec \theta)
    &:= \frac{\sqrt{\sum_{\myvec m \in \mathcal{M}}(1-\delta_{\Tilde{x}_i(\myvec m),-1})\sum_{\myvec m \in \mathcal{M}}(1-\delta_{\Tilde{x}_j(\myvec m),-1})}}{\sqrt{\sum_{\myvec m \in \mathcal{M}}(1-\delta_{\Tilde{x}_i(\myvec m),-1})\sum_{\myvec m \in \mathcal{M}}(1-\delta_{\Tilde{x}_j(\myvec m),-1})}+\sum_{\myvec m \in \mathcal{M}}(1-\delta_{\Tilde{x}_i(\myvec m),-1})(1-\delta_{\Tilde{x}_j(\myvec m),-1})}
    \label{eq:estimator_mu}
\end{align}


Note, that for disjoint coverings $\mathcal{A}$, the sums over $r$ collapse to at most a single term as every bit $i$ is contained in exactly one register-set $A_{r(i)}$. In the definition of $\Tilde{q}_{ij}$ we adopt the convention $\Tilde{q}_{ij}\equiv 0$ if $\{r: i,j \in A_{r}\}=\emptyset$ (sum contains no terms).

\subsection{Derivation for disjoint covering}\label{subapp:derivation_disj_covering_estimators}

To derive the marginal probability distributions $p_{ij}(\myvec \theta) = \text{Prob}_{\myvec \theta}(x_i = 1, x_j = 1)$ and $p_i(\myvec \theta) = \text{Prob}_{\myvec \theta}(x_i = 1)$ that arise from the greedy sampling algorithm
, consider the output of the PQC given by equation~\ref{eq:quantum_state_qubit_compression}.

We denote the sequence of registers during one run of the sampling algorithm by $S = \left (r_i \right)_{i=1}^{N_{m}}$, where $r_i$ is the register sampled in the $i^{\text{th}}$ measurement (omitting measurements with not bits being set)
and $N_m$ is the number of such calls before termination of the algorithm. Denote by $\Tilde{A}^{(S)}_{r_i}\subset A_{r_i}$\footnote{In the case of a disjoint covering $\mathcal{A}$, we have $\Tilde{A}^{(S)}_{r_i}= A_{r_i}$} the bits set in the $i^{\text{th}}$ measurement.

Then 
\begin{equation}
    \text{Prob}(\myvec x) = \sum_{S} \text{Prob}(S) \prod_{i = 1}^{N_m(S)} p^{(r_i)}_{\Tilde{A}^{(S)}_i}(\myvec x)
\end{equation}

where $\text{Prob}(S)$ is the probability of the sampling algorithm resulting in the register-sequence $S$ and $p^{(r_i)}_{\Tilde{A}^{(S)}_{r_i}}(\myvec x)$ is defined through the complex ancilla amplitudes compatible with $\myvec x_{\Tilde{A}^{(S)}_i}$ for the given register $r_i$, i.e. 
\[
    p^{(r_i)}_{\Tilde{A}^{(S)}_i}(\myvec x) = \sum_{b_k 
    \begin{cases}
        \in\{0,1\}, &\text{for } A_{r_i}[j] \notin \Tilde{A}^{(S)}_{r_i}\\
        =x_{A_{r_i}[k]}, &\text{for } A_{r_i}[k] \in \Tilde{A}^{(S)}_{r_i}
    \end{cases}}|a^{b_1 ... b_{n_a}}_{r}|^2
\]
which simplifies to $|a^{\myvec x_{A_r}}|^2$ if $|\Tilde{A}^{(S)}_{r_i}|=n_a$.

For fixed $S$ define $\Tilde{r}(i),~i\in B$ such that $i\in \Tilde{A}^{(S)}_{\Tilde{r}(i)}$, as well as $p_{i}^{(r)} := p_{\{i\}}^{(r)}(\myvec 1) = \sum_{\substack{b_k \in \{0,1\}\\b_{l_r(i)}=1}}|a^{b_1 ... b_{n_a}}_{r}|^2$ and $p_{ij}^{(r)} := p_{\{i,j\}}^{(r)}(\myvec 1) = \sum_{\substack{b_k \in \{0,1\}\\b_{l_r(i)}=b_{l_r(j)}=1}}|a^{b_1 ... b_{n_a}}_{r}|^2$. Then it follows for the marginals:
\begin{equation}
    p_i = \sum_{\myvec x : x_i = 1} \text{Prob}(\myvec x) = \sum_{S} \text{Prob}(S)p_{i}^{(\Tilde{r}(i))}
\end{equation}

\begin{equation}
    p_{ij} = \sum_{\myvec x :~x_i = 1, x_j = 1} \text{Prob}(\myvec x) = \sum_{S:~\Tilde{r}(i)\neq \Tilde{r}(j)} \text{Prob}(S) p_{i}^{(\Tilde{r}(i))}p_{j}^{(\Tilde{r}(j))} + \sum_{S:~~\Tilde{r}(i)= \Tilde{r}(j)} \text{Prob}(S) p_{ij}^{(\Tilde{r}(i))}
\end{equation}

In the case of a disjoint covering $\mathcal{A}$, every register has to be sampled once and $\Tilde{A}^{(S)}_{r_i}=A_{r_i}~\forall i\in \{1,\ldots,N_m=\frac{I}{n_a}\}$. Due to this the sum over $S$ merely consists of different orderings of $\{1,\ldots,N_r=\frac{I}{n_a}\}$. Therefore, $\text{Prob}(S)=\prod_{i = 1}^{N_m}|\beta_{r_i}|^2=(\frac{I}{n_a}!)^{-1}$ is uniform.

For general coverings however, $\text{Prob}(S)$ is not necessarily uniform and the sum over $S$ is non-trivial. However, we can \emph{hope to approximate} in the general case
\begin{equation}
    p_i = \sum_{S} \text{Prob}(S)p_{i}^{(r(i))} \approx \frac{1}{\sum_{r:~i \in A_r}|\beta_r|^2}\sum_{r:~i \in A_r}|\beta_r|^2 p_{i}^{(r)}
\end{equation}
and

\begin{align}
    p_{ij} &= \sum_{S:~r(i)\neq r(j)} \text{Prob}(S) p_{i}^{(r(i))}p_{j}^{(r(j))} + \sum_{S:~~r(i)= r(j)} \text{Prob}(S) p_{ij}^{(r(i))}\\
    &\approx \mu_{ij} \frac{1}{\sum_{r_1, r_2:~i \in A_{r_1}, j \notin A_{r_1}, j \in A_{r_2}}|\beta_{r_1}|^2|\beta_{r_2}|^2}\sum_{r_1, r_2:~i \in A_{r_1}, j \notin A_{r_1}, j \in A_{r_2}}|\beta_{r_1}|^2|\beta_{r_2}|^2p_{i}^{(r_1)}p_{j}^{(r_2)}\\ &+ (1-\mu_{ij})\frac{1}{\sum_{r:~i, j \in A_r}|\beta_r|^2}\sum_{r:~i, j \in A_r} |\beta_r|^2 p_{ij}^{(r)}
\end{align}

(equality holds for disjoint coverings, in which case the sums become trivial). We can regard this as assuming that based on symmetry considerations\footnote{which depending on the graph covering $\mathcal{A}$ may not be warranted}, the probability of the sampling algorithm running through the sequence $S$ where entry $i$ is sampled from register $r\in \{r:~i\in A_r\}$) is proportional to $|\beta_r|^2$. Similar considerations are made when looking at bit-pairs $(i,j)$ with the added complexity that they can either be sampled from the same register $r$ or from two different register $r_1, r_2$. The probability of the former is estimated as
\begin{equation}
    \mu_{ij} \approx \frac{\sqrt{\sum_{r_1, r_2:~i \in A_{r_1}, j \notin A_{r_1}, j \in A_{r_2}}|\beta_{r_1}|^2|\beta_{r_2}|^2}}{\sqrt{\sum_{r_1, r_2:~i \in A_{r_1}, j \notin A_{r_1}, j \in A_{r_2}}|\beta_{r_1}|^2|\beta_{r_2}|^2}+\sum_{r:~i, j \in A_r}|\beta_r|^2}.
\end{equation}
By further approximating $\sum_{r_1, r_2:~i \in A_{r_1}, j \notin A_{r_1}, j \in A_{r_2}}\approx\sum_{r_1:~i \in A_{r_1}}\sum_{r_2:~j \in A_{r_2}}$ we obtain
\begin{equation}
    p_{ij} \approx \mu_{ij}\sum_{r_1:~i \in A_{r_1}}\sum_{r_2:~j \in A_{r_2}}p_i^{(r_1)} p_j^{(r_2)}+ (1-\mu_{ij})\frac{1}{\sum_{r:~i, j \in A_r}|\beta_r|^2}\sum_{r:~i, j \in r} |\beta_r|^2 p_{ij}^{(r)}.
\end{equation}
Given this, $p_i$ is estimated by $\hat{p}_i$ (equ.~\ref{eq:p_hat_pauli}) and $p_{ij}$ by $\hat{p}_{ij}$ (equ.~\ref{eq:p_ij_hat_pauli}), where all \enquote{$\approx$} are exact for the case of a disjoint covering.

In summary, we have motivated the cost estimator in section~\ref{subsec:cost} as a heuristic for the general case and proven 

\begin{restatable}[]{lm}{lemmacostestimator}
\label{lemma:cost_estimator}
    For the greedy sampling algorithm
     and a disjoint covering $\mathcal{A}$ (perfect matching), equation~\ref{eq:cost_estimator} and $\mathbb{E}_{\myvec \theta}[C]$ from equation \ref{eq:QUBO_as_expectation_min} are equal in the limit $n_{\text{shots}}\to \infty$.
    In particular
    \begin{align}
        p_{i}(\myvec \theta) &= \sum_{\substack{b_k \in \{0,1\}\\b_{l_r(i)}=1}}|a^{b_1 ... b_{n_a}}_{r}|^2 \text{ where $r$ s.t. } i\in A_r\\
        p_{ij}(\myvec \theta) &=
        \begin{cases}
             \sum_{\substack{b_k \in \{0,1\}\\b_{l_{r_i}(i)}=b_{l_{r_i}(j)}=1}}|a^{b_1 ... b_{n_a}}_{r_i}|^2, &\text{if }r_i = r_j \\
             \sum_{\substack{b_k \in \{0,1\}\\b_{l_{r_i}(i)}=1}}|a^{b_1 ... b_{n_a}}_{r_i}|^2 \sum_{\substack{b_k \in \{0,1\}\\b_{l_{r_j}(j)}=1}}|a^{b_1 ... b_{n_a}}_{r_j}|^2, &\text{if }r_i \neq r_j 
        \end{cases}
        \text{  where $r_x$ s.t. } x\in A_{r_x}
    \end{align}
\end{restatable}

\subsection{Explicit form of cost-gradient} \label{subapp:gradient_cost}
In this section, we will give the explicit form of the derivatives $\partial_{\theta_d}$ of the cost a and register-regularization estimator

\begin{equation}
    \hat{C}(\myvec \theta) + \hat{R}(\myvec \theta)  = \sum_{\substack{i,j = 1\\i\neq j}}^I \hat{p}_{ij}(\myvec \theta) A_{ij} + \sum_{i = 1}^I \hat{p}_{i}(\myvec \theta) (A_{ii}+b_i(\hat{\myvec s}(\myvec \theta))) + c(\hat{\myvec s}(\myvec \theta))
    +\eta \sum_{r = 1}^{N_r}\left [\hat{r}_r(\myvec \theta) -\frac{1}{N_r} \right]^2.
    \label{eq:cost_estimator_and_regularization}
\end{equation}

When optimizing $\hat{C}(\myvec \theta)$ with respect to $\myvec \theta$ we may want to make use of the gradient $\nabla_{\mthet}\hat{C}$ to update $\myvec \theta$. For this, note that both terms in the quotient for $\hat{p}_i$ and $\hat{q}_{ij}$ are given as linear combinations of Pauli-expectation values (equation~\ref{eq:p_hat_pauli} and \ref{eq:q_hat_pauli}). For many variational ansätze, in particular those consisting of single-qubit Pauli-rotations used in this paper, this allows to easily calculate gradients $\nabla_{\mthet}[\hat{p}_i]$ and $\nabla_{\mthet}[\hat{q}_{ij}]$ through the parameter-shift rule (\cite{mitarai_quantum_2018, schuld_evaluating_2019}). Similarly, the gradients of $\hat{\mu}_{ij}$ and $\hat{r}_r$ can be calculated by applying the parameter-shift rule.
Through application of the chain-rule, the gradient of $\hat{C}\left(\{\hat{p}_i\}_i, \{\hat{q}_{ij}\}_{ij}, \{\hat{\mu}_{ij}\}_{ij};\{\hat{r}_r\}_{r}\right )$ is obtained.

We distinguish between the case of register-preserving ansätze and general circuit ansätze. For simplicity (and motivated by the regularization of the register-probabilities) we will treat $\hat{\mu}_{ij}$ as constant in $\myvec\theta$ in both cases.

By the chain rule we get
\begin{align}
    \partial_{\theta_d}\hat{C}(\myvec \theta) &= \sum_{\substack{i,j = 1\\i\neq j}}^I  \partial_{\theta_d}\hat{p}_{ij} A_{ij} + \sum_{i = 1}^I\left [   \partial_{\theta_d}\hat{p}_{i} (A_{ii}+b_i(\hat{\myvec s}))+\hat{p}_{i}\nabla_{\myvec s}b_i^{T}(\hat{\myvec s})\partial_{\theta_d}\hat{\myvec s}\right ] + \nabla_{\myvec s}c(\hat{\myvec s})\partial_{\theta_d}\hat{\myvec s} \\
    \partial_{\theta_d}\hat{R}(\myvec \theta) &= 2\eta \sum_{r = 1}^{N_r}\left [\hat{r}_r -\frac{1}{N_r} \right]\partial_{\theta_d}\hat{r}_r
\end{align}
where
\begin{align}
    \nabla_{\myvec s}b_i &= 2\lambda V_{(i)} \\
    \nabla_{\myvec s}c &= -2\lambda \left(\myvec l+\myvec s - \myvec{bal}\right)\\
    \partial_{\theta_d}\hat{\myvec s} &= \text{Diag}[(\delta_{\hat{s}_i> 0})_{1\leq i\leq KJ}]\sum_{i=1}^I\partial_{\theta_d}\hat{p}_{i} V_{(i)}\\
    \partial_{\theta_d} \hat{p}_{ij} &= (1-\hat{\mu}_{ij}) \partial_{\theta_d} \hat{q}_{ij}+\hat{\mu}_{ij}( \hat{p}_j\partial_{\theta_d}\hat{p}_i+\hat{p}_i\partial_{\theta_d}\hat{p}_j)
\end{align}
with $V_{(i)}$ being the $i^\text{th}$ row of V.
As $\hat{r}_r$ is the expectation of the observable $\mathds{1}_{n_a}\otimes \ket{r}\bra{r}$ its partial derivatives can be calculated directly through parameter-shift rules or similar techniques to calculate the gradient of quantum observables. If we consider a register-preserving ansatz, the same holds true for $\hat{q}_{ij}(\myvec \theta)$ and $\hat{p}_{i}(\myvec \theta)$. Otherwise, the derivatives have to be calculated separately for the nominator and denominator in equation~\ref{eq:p_hat_pauli} and \ref{eq:q_hat_pauli} respectively and recombined using the quotient rule.

\section{Expressing cost estimator as Hermitian observable}\label{app:cost_as_observable}
We observed in section~\ref{subsec:qubit_compression} and \ref{subsec:cost}, that our qubit-compression results in a cost-estimator expressed as a function of Pauli-Z measurements that cannot be written straightforwardly as the expectation over a Hermitian operator.
Here, we alleviate this issue for register-preserving circuits and fixed slack variables.




We can rewrite the cost as a Hermitian expectation for register-preserving circuits and fixed slack variables $\myvec s$:
\begin{enumerate}
    \item Substitute denominators in $\hat{p}_i$ and $\hat{q}_{ij}$ as well as $\hat{\mu}_{ij}$ by exact counterpart.
    \item Define operator $C(\myvec \theta)$ such that $\hat{C}(\myvec \theta) = \langle C(\myvec \theta)\rangle_{\mathcal{M}}$ by \enquote{doubling the Hilbert space}.
\end{enumerate}

\paragraph{1.} The denominators in the expressions for $\hat{p}_i$ and $\hat{q}_{ij}$ (equations \ref{eq:p_hat_pauli} and \ref{eq:q_hat_pauli}) as well as $\hat{\mu}_{ij}$ can be replaced by scalar constants by substituting 
\begin{equation}
    \langle \mathds{1}_{n_a}\otimes\ket{r}\bra{r}\rangle_{\mathcal{M}}\mapsto\frac{1}{N_r}
\end{equation}
resulting in ($n_i:=\sum_{\substack{r=1\\i\in A_r}}^{N_r}1$, $n_{ij}:=\sum_{\substack{r=1\\i,j\in A_r}}^{N_r}1$) 
\begin{align}
    \hat{\mu}_{ij}&=\frac{\sqrt{n_i n_j}}{\sqrt{n_i n_j}+n_{ij}}\\
    \hat{p}_i&=\frac{N_r}{n_i}\left\langle \sum_{\substack{r = 1 \\ i \in A_r}}^{N_r}\mathds{1}_{l_r(i)-1}\otimes\ket{1}\bra{1}\otimes \mathds{1}_{n_a-l_r(i)}\otimes\ket{r}\bra{r}\right\rangle_{\mathcal{M}}\label{eq:p_hermitian_exp}\\
    \hat{q}_{ij}&=\frac{N_r}{n_{ij}}\left\langle \sum_{\substack{r = 1 \\ i,j \in A_r}}^{N_r}\mathds{1}_{\min\{l_r(i),l_r(j)\}-1}\otimes\ket{1}\bra{1}\otimes\mathds{1}_{|l_r(i)-l_r(j)|-1}\otimes\ket{1}\bra{1}\otimes \mathds{1}_{n_a-l_r(i)-l_r(j)}\otimes\ket{r}\bra{r}\right\rangle_{\mathcal{M}}\label{eq:q_hermitian_exp}
\end{align}
where we set $\hat{p}_i=0$ ($\hat{q}_{ij}=0$) if $n_i = 0$ ($n_{ij}=0$).

\paragraph{2.} Equations \ref{eq:p_hermitian_exp} and \ref{eq:q_hermitian_exp} suggest defining the Hermitian operators
\begin{align}
    P_i&=\frac{N_r}{n_i}\sum_{\substack{r = 1 \\ i \in A_r}}^{N_r}\mathds{1}_{l_r(i)-1}\otimes\ket{1}\bra{1}\otimes \mathds{1}_{n_a-l_r(i)}\otimes\ket{r}\bra{r}\\
    Q_{ij}&=\frac{N_r}{n_{ij}} \sum_{\substack{r = 1 \\ i,j \in A_r}}^{N_r}\mathds{1}_{\min\{l_r(i),l_r(j)\}-1}\otimes\ket{1}\bra{1}\otimes\mathds{1}_{|l_r(i)-l_r(j)|-1}\otimes\ket{1}\bra{1}\otimes \mathds{1}_{n_a-l_r(i)-l_r(j)}\otimes\ket{r}\bra{r}.
\end{align}
This would allow us to write $\hat{C}(\myvec \theta)$ in the desired form if it weren't for the terms of the form $\hat{p}_i\hat{p}_j = \left\langle P_i\right\rangle_{\mathcal{M}}\left\langle P_j\right\rangle_{\mathcal{M}}$. Instead of a quantum state $\ket{\psi({\myvec \theta})}$ in the form of equation~\ref{eq:quantum_state_qubit_compression}, we consider the product state $\ket{\psi({\myvec \theta})}\otimes\ket{\psi({\myvec \theta})}$ and the following operator acting on it:
\begin{equation}
    C = .\sum_{i,j=1}^{I}A_{ij}\left[(1-\hat{\mu}_{ij})Q_{ij}\otimes\mathds{1}_{n_q}+\hat{\mu}_{ij} P_i\otimes P_j\right]+\sum_{i=1}^I (A_{ii}+b_i(\myvec s)) P_i\otimes\mathds{1}_{n_q}+ c(\myvec s)\mathds{1}_{2n_q}
    \label{eq:C_obs}
\end{equation}
This requires doubling the number of qubits and circuit width.
The cost estimator is then given as this operator's expectation value
\begin{equation}
    \hat{C}(\myvec \theta) = \langle  C \rangle_{\mathcal{M}} \simeq \bra{\psi({\myvec \theta})}\otimes\bra{\psi({\myvec \theta})} C \ket{\psi({\myvec \theta})}\otimes\ket{\psi({\myvec \theta})}
\end{equation}
estimated through measurements.

We note:
\begin{itemize}
    \item C is Hermitian as it is real-valued and diagonal in the computational basis
    \item While we assumed the measured quantum state to be register-uniform in the substitutions of 1., one may consider using the observable in equation~\ref{eq:C_obs} even for non-register-preserving circuits if the register-amplitudes are approximately kept constant with a penalty term. 
    The penalty term can be added to $C$ in the same manner as $P_i\otimes P_j$.
    \item The assumption of a register-preserving circuit comes with a caveat: We need to be careful in applying methods tailored to Hermitian-expectation-minimization which change the variational ansatz itself based on properties of the Hermitian (QAOA being the most prominent example).
    Still there are relevant results which can be applied to a constrained set of allowed gates, such as optimization techniques (\cite{schuld_evaluating_2019, ostaszewski_structure_2021}), estimation (\cite{huang_predicting_2020}) and error mitigation techniques as well as libraries (\cite{obrien_error_2021, endo_practical_2018, larose_mitiq_2022}) or fault-tolerant methods for evaluating expectation values (e.g.~\cite{knill_optimal_2007}). 
    \item Throughout, we assumed constant slack variables $\myvec s$, as expressing the relationship in~\ref{eq:slack_estimator} through an observable is complicated by the non-linearity of the rectified linear unit. Hence, alternating adjustments of $\myvec s$ using~\ref{eq:slack_estimator} and the circuit parameters $\myvec \theta$ through a classical optimizer are needed. Alternatively, both variables could be optimized simultaneously, for which implicit differentiation may be useful (\cite{ahmed_implicit_2022}).
\end{itemize}


\section{Simulation parameters}\label{app:simulation_param}

The configurations used to determine ansatz and for optimizing circuit parameters in section~\ref{sec:results} are shown in table~\ref{tab:setup_joined_results_bm}. Pennylane (\cite{bergholm_pennylane_2022}; version:~\texttt{pennylane=0.29.1}) was used for quantum computing simulations. For parameter optimization, the SciPy (\cite{virtanen_scipy_2020}; version:~\texttt{scipy=1.10.1}) implementation of gradient-free optimizer COBYLA as well as standard gradient descent were used. The gradients of the latter were calculated through the chain rule and parameter-shift rule (\cite{mitarai_quantum_2018, schuld_evaluating_2019}).

\begin{table}
    \caption[Benchmarking setup]{Benchmarking setup} 
    \centering
    \begin{tabular}{@{}l r@{}}
        \toprule
        {\bfseries Parameter} & {\bfseries Value(s)} \\
        \midrule
        Cost penalty $\lambda$ & 10\\ 
        Register-regularization penalty$^*$ $\eta$  & 1000\\
        Gradient max steps & 1500\\ 
        Gradient stepsize & $2.5\times 10^{-4}$\\ 
        Optimizer & COBYLA, Gradient descent (DESC)\\ 
        Number of runs / starting point & 25\\
        Parameter initialization & Uniform random $\in [0,2\pi]$ \\
        Depth $d$ & 1, 4\\
        Ancilla qubits $n_a$& 1, 4, 8, 16\\
        $n_{\text{shots}}$ & $10^4$, $2\times 10^4$\\
        Ansatz & Register-preserving, Hardware-efficient\\
        Register-mapping & Disjoint covering unless $n_r+\geq1$ \\
        \bottomrule
        $^*$ For Hardware-efficient ansatz only
    \end{tabular}
    \label{tab:setup_joined_results_bm}
\end{table}

\end{document}